\documentclass{lmcs}
\pdfoutput=1

\usepackage{lastpage}
\lmcsdoi{20}{4}{20}
\lmcsheading{}{\pageref{LastPage}}{}{}%
{Sep.~26,~2023}{Dec.~03,~2024}{}

\usepackage[utf8]{inputenc} 
\usepackage{amssymb}
\usepackage{mathpartir}
\usepackage{url}
\usepackage{amsmath}
\usepackage{stmaryrd}
\usepackage{tikz}
\usetikzlibrary{calc,arrows,arrows.meta,shapes,snakes,automata,backgrounds,petri,
  decorations.pathmorphing,positioning,fit} \usepackage{caption}
\usepackage{subcaption} \usepackage{xspace}

\usepackage[textwidth=3cm]{todonotes}
\usepackage{xcolor}
\usepackage{listings}

\usepackage{hyperref}
\usepackage[capitalise]{cleveref}
\usepackage[normalem]{ulem}
\usepackage[subnum]{cases}
\newcounter{sarrow}

\newcommand{\cammino}[1]{\mathtt{path}(#1)}
\newcommand{\fuse}{\bowtie}

\newcommand{\names}{\mathtt{n}}
\newcommand{\bn}{\mathtt{bn}}

\newcommand{\aeq}{=_{\alpha}}

\newcommand{\lts}[1]{\xrightarrow{#1}}
\newcommand{\para}{\parallel}

\newcommand{\mem}[4]{\langle{#1},{#2},{#3}\rangle\cdot{#4}}

\newcommand{\proc}{\triangleright}
\newcommand{\con}{\equiv}
\newcommand{\precon}{\preceq}

\newcommand{\rpar}[1]{\langle 1 \rangle \cdot{#1}}
\newcommand{\lpar}[1]{\langle  0 \rangle \cdot{#1}}

\newcommand{\emp}{\langle\rangle}
\newcommand{\ltsk}[2]{\xrightarrow{{#1}:{#2}}}

\newcommand{\rccst}[1]{\xhookrightarrow{#1}}

\newcommand{\bk}{\rightsquigarrow}
\newcommand{\fw}{\rightarrow}

\newcommand{\sdef}{::=}
\newcommand{\sdefco}{::=^\mathtt{co}}

\newcommand{\ou}{\;\,|\;\,}
\newcommand{\nil}{\mathbf{0}}
\newcommand{\co}[1]{\overline{#1}}
\newcommand{\rev}{^\bullet}

\newcommand{\event}[3]{\langle #1,#2,#3\rangle }

\newcommand{\hide}{\backslash}

\newcommand{\entry}[1]{\langle #1 \rangle}
\newcommand{\st}{\,\mid\,}

\newcommand{\procs}{\mathcal{P}}
\newcommand{\procr}{\mathcal{P}_R}

\newcommand{\ctx}[2]{#1[#2]}
\newcommand{\fn}{\mathtt{fn}}

\newcommand{\trans}[1]{\ensuremath{\,[\/{#1}\/\rangle}\,}
\newcommand{\pre}[1]{\ensuremath{\!\!~^{\bullet}{#1}}}
\newcommand{\post}[1]{\ensuremath{{#1}{^{\bullet}}}}

\newcommand{\flt}[1]{\ensuremath{[\![{#1}]\!]}}
\newcommand{\subnet}[3]{\ensuremath{#1|_{#2}^{#3}}}

\newcommand{\setenum}[1]{\{#1\}}
\newcommand{\setcomp}[2]{\{{#1} \mid {#2}\}}
\newcommand{\card}[1]{\ensuremath{|#1|}}
\newcommand{\reachMark}[1]{\ensuremath{\mathcal{M}_{#1}}}

\newcommand{\firseq}[2]{\ensuremath{\mathcal{R}^{#1}_{#2}}}
\newcommand{\states}[1]{\ensuremath{\mathbb{E}(#1)}}

\newcommand{\lead}[1]{\ensuremath{\mathit{lead}(#1)}}
\newcommand{\start}[1]{\ensuremath{\mathit{start}(#1)}}
\newcommand{\fs}{\textsf{fs}}

\newcommand{\remains}[1]{\ensuremath{\mathit{tail}(#1)}}
\newcommand{\MC}[1]{\ensuremath{\sim}}

\newcommand{\marko}[1]{\ensuremath{\mathsf{#1}}}

\newcommand{\ca}{\textsf{CA}}
\newcommand{\unet}{\textsf{UN}}

\newcommand{\pastdec}[1]{\hat #1}

\newcommand{\ccstonet}[1]{\mathcal{N}(#1)}

\newcommand{\zeronetBC}{\ccstonet \nil}
\newcommand{\prefnetBC}[2]{\ccstonet {#1.#2}}
\newcommand{\netBCsimple}[1]{\ccstonet {#1}}
\newcommand{\parnetBC}[2]{\ccstonet {#1 \| #2}}

\newcommand{\choicenetBCm}[2]{+_{#1} #2}
\newcommand{\restnetBC}[2]{\ccstonet {#1\setminus#2}}

\newcommand{\pardec}[1]{{\sf\|_{#1}}}
\newcommand{\plusdec}[1]{{\sf +_{#1}}}
\newcommand{\restdec}[1]{{\sf \setminus_{#1}}}

\newcommand{\lbl}[1]{\ell(#1)}

\newcommand{\newinenc}[1]{\textcolor{magenta}{#1}}

\newcommand{\ancestor}[1]{{\rho}(#1)}
\newcommand{\anc}{\rightharpoondown}

\newcommand{\mrkof}[1]{\mathbf{\mu}({#1})}

\makeatletter
\DeclareRobustCommand{\cev}[1]{%
  \mathpalette\do@cev{#1}%
}
\newcommand{\do@cev}[2]{%
  \fix@cev{#1}{+}%
  \reflectbox{$\m@th#1\overrightarrow{\reflectbox{$\fix@cev{#1}{-}\m@th#1#2\fix@cev{#1}{+}$}}$}%
  \fix@cev{#1}{-}%
}
\newcommand{\fix@cev}[2]{%
  \ifx#1\displaystyle%
    \mkern#23mu
  \else
    \ifx#1\textstyle%
      \mkern#23mu
    \else
      \ifx#1\scriptstyle%
        \mkern#22mu
      \else
        \mkern#22mu
      \fi
    \fi
  \fi
}

\makeatother

\newcommand{\revnet}[1]{\cev{#1}}

\newcommand{\CCS}{CCS\xspace}
\newcommand{\Haskell}{{\sf Haskell}\xspace}

\newcommand\revk[1]{%
\stepcounter{sarrow}%
\mathrel{\begin{tikzpicture}[baseline= {( $ (current bounding box.south) + (0,-0.5ex) $ )}]
\node[inner sep=.5ex] (\thesarrow) {$\scriptstyle #1$};
\path[draw,<-,decorate,
  decoration={zigzag,amplitude=1.2pt,segment length=1.5mm,pre=lineto,pre length=6pt}]
    (\thesarrow.south east) -- (\thesarrow.south west);
\end{tikzpicture}}%
}
\newcommand{\rltsk}[2]{\revk{{#1}:{#2}}}

\makeatletter
\providecommand*\ltsc[2][]{%
  \ext@arrow 0055{\arrowfill@\relbar\relbar\rightarrowtriangle}{#1}{#2}}
\makeatother

\definecolor{mymagenta}{rgb}{0.5,0,0.5}
\definecolor{myred}{rgb}{0.5,0,0}
\definecolor{mygreen}{rgb}{0,0.4,0}
\definecolor{myblue}{rgb}{0,0,0.6}
\definecolor{lightblue}{rgb}{0,0.5,1}
\definecolor{mygray}{gray}{0.5}
\definecolor{lightgray}{gray}{0.95}

\newcommand{\macronet}[2]{\ensuremath{\langle S_{#1}{#2}, T_{#1}{#2}, F_{#1}{#2}, \marko{m}_{#1}{#2}\rangle}}
\newcommand{\macronetbis}[3]{\ensuremath{\langle S_{#1}{#2}, T_{#1}{#2}, F_{#1}{#2}, {#3}\rangle}}
\newcommand{\zero}{\ensuremath{\mathsf{0}}}
\newcommand{\Er}{U}
\newcommand{\er}{u}
\newcommand{\fe}{\mathtt{f}}
\newcommand{\re}{\mathtt{r}}

\newcommand{\macrorevnet}[2]{\ensuremath{\langle S_{#1}{#2}, T_{#1}{#2}, 
\Er_{#1}{#2}, F_{#1}{#2}, \marko{m}_{#1}{#2}\rangle}}

\newcommand{\rcn}{{\sf rUN}}

\newcommand{\bundle}{\ensuremath{\mapsto}}

\newcommand{\un}[1]{\underline{#1}}

\newcommand{\muset}[1]{\ensuremath{\partial #1}} 

\newcommand{\Comment}[1]{}

\newcommand{\revunr}[2]{\revnet{#1}^{#2}}

\newcommand{\keypl}[1]{\ensuremath{\mathcal{K}_{#1}}}

\usepackage{listings}
\lstdefinestyle{INLINE}{
}
\lstdefinestyle{DISPLAY}{
  numberstyle=\tiny\tt\color{mygray},
  numbersep=1em,
  xleftmargin=4mm, 
  backgroundcolor=\color{white},
}
\lstdefinestyle{FLOAT}{
  float,
  captionpos=b,
}
\lstnewenvironment{code}
    {\lstset{}%
      \csname lst@SetFirstLabel\endcsname}
    {\csname lst@SaveFirstLabel\endcsname}
\lstnewenvironment{HaskellF}[1][]{\lstset{language=Haskell,style=DISPLAY,style=FLOAT,#1}}{}
\lstnewenvironment{HaskellD}[1][]{\lstset{language=Haskell,style=DISPLAY,#1}}{}
\newcommand{\HI}{\lstinline[language=Haskell,style=INLINE]}

\newcommand{\InputHaskellD}[4][]{%
  \lstinputlisting[
    language=Haskell,
    linerange=#2-#3,
    style=DISPLAY,#1]{./src/#4}%
}
\lstnewenvironment{PIC}[1][]{\lstset{language=PIC,style=DISPLAY,#1}}{}

\crefname{defi}{Definition}{Definitions}
\Crefname{defi}{Definition}{Definitions}
\crefname{cor}{Corollary}{Corollaries}
\Crefname{cor}{Corollary}{Corollaries}
\crefname{lem}{Lemma}{Lemma}
\Crefname{lem}{Lemma}{Lemma}

\begin{document}
\title{A truly concurrent semantics for reversible CCS}
\thanks{This work has
  been partially supported by the BehAPI project funded by the EU H2020 RISE
  under the Marie Sklodowska-Curie action (No: 778233), by the Italian PRIN
  2020 project NiRvAna -- Noninterference and Reversibility Analysis in
  Private Blockchains, 
  the Italian Ministry of Education, University and Research through the PRIN 2022 project `` Developing Kleene Logics and their Application'' (DeKLA), project code: 2022SM4XC8, 
  the INdAM-GNCS E53C22001930001 project
  RISICO -- Reversibilit\`a in Sistemi Concorrenti: Analisi Quantitative e
  Funzionali, and the European Union - NextGenerationEU SEcurity and RIghts
  in the CyberSpace (SERICS) Research and Innovation Program PE00000014,
  projects STRIDE and SWOP}

\author[H. Melgratti]{Hern\'an Melgratti\lmcsorcid{0000-0003-0760-0618}}[a]

\author[C. A. Mezzina]{Claudio Antares
  Mezzina\lmcsorcid{0000-0003-1556-2623}}[b]

\author[G. M. Pinna]{G. Michele Pinna\lmcsorcid{0000-0001-8911-1580}}[c]

\address{ICC - Universidad de Buenos Aires - Conicet, Argentina}
\address{Dipartimento di Scienze Pure e Applicate, Universit\`a di Urbino, Italy}
\address{ Universit\`a di Cagliari, Italy}

%
%

\begin{abstract}
  Reversible CCS (RCCS) is a well-established, formal model for reversible
  communicating systems, which has been built on top of the classical Calculus
  of Communicating Systems (CCS). In its original formulation, each CCS
  process is equipped with a memory that records its performed actions, which
  is then used to reverse computations. More recently, abstract models for
  RCCS have been proposed in the literature, essentialy, by directly
  associating RCCS processes with (reversible versions of) event structures.
  In this paper
we propose a different abstract model: starting from one of the well-known
  encoding of CCS into Petri nets we apply a recently proposed approach to
  incorporate causally-consistent reversibility to Petri nets, obtaining as
  result the (reversible) net counterpart of every RCCS term.

\end{abstract}

\maketitle

\section{Introduction}\label{sec:intro}
The calculus for concurrent systems (CCS)~\cite{Milbook} serves as one of the
foundational frameworks for concurrent systems. Typically, these systems are
described as the parallel composition of \emph{processes} (also referred to as
components), which interact by sending and receiving messages through named
channels.
Processes are defined in terms of communication actions performed over
specific channels. For example, we use $a$ and $\co a$ to respectively
represent a receive and a send action over the channel $a$.
Basic actions can be combined using prefixing ($\_ . \_$), choice ($\_ + \_$),
and parallel ($\_\parallel\_$) operators.
Initially, the semantics of CCS were based on the \emph{interleaved} approach,
which considers only executions that arise from a single processor.
Consequently, parallelism was reduced to non-deterministic choices and
prefixing.
For instance, under the interleaved approach, the CCS processes
$a \parallel b$ and $a.b\ +\ b.a$ are considered \emph{equivalent}. This means
that the framework does not distinguish between a process that can perform
actions $a$ and $b$ concurrently and one that sequentially executes these
actions in any possible order (interleaving/schedule).
To address this limitation, subsequent research aimed to equip CCS with
\emph{true concurrent} semantics, adopting styles similar to Petri
nets~\cite{ReiBook} and Event Structures~\cite{Win88,NPW:PNES}.
It has been shown that every CCS process can be associated with a
corresponding Petri net that can mimic its computations.
Various flavors of Petri nets have been explored in the literature, including
\emph{occurrence} nets~\cite{Goltz90}, a variant of \emph{Conditions/Events}
nets~\cite{DNMActa88}, and \emph{flow} nets~\cite{BCIC94}.
The works in~\cite{Win:ES} and~\cite{BCIC94} have additionally shown that the
computation of a CCS process can be represented by  event structures.

In the last decades, many efforts were made to endow computation models with
reversible semantics~\cite{wg1,wg2}. In particular, two different models have
been proposed for CCS: reversible CCS (RCCS)~\cite{rccs,rccsnew} and CCS with communication keys (CCSK)~\cite{ccsk}.
Both of them incorporate a logging mechanism in the operational semantics of
CCS that enables the undoing of computation steps. Moreover, it has been shown
that they are isomorphic~\cite{LaneseMM21} since they only differ on how they
log information about past computations: while RCCS relies on some form of
\emph{memory/monitor}, CCSK uses \emph{keys}.
Previous approaches have also developed true concurrent semantics for
reversible versions of CCS. For instance, it has been shown that CCSK can be
associated with reversible bundle event
structures~\cite{GraversenPY18,GraversenPY2021}.
Also configuration structures have been associated to RCCS \cite{AubertC20}.
Nonetheless, we still lack a Petri net model for reversible CCS processes.
We may exploit some recent results that connect reversible occurrence nets
with reversible event
structures~\cite{MelgrattiMU20,MelgrattiM0PU20,MelgrattiMP21,tocl} to indirectly
recover a Petri net model from the reversible bundle event structures defined
in~\cite{GraversenPY18}.
However, we follow a different approach, which is  more direct:
\begin{enumerate}
\item We encode CCS processes into a mild generalization of occurrence nets,
  namely \emph{unravel nets}, in the vein of Boudol and
  Castellani~\cite{BCIC94}.
\item We show that \emph{unravel nets} can be made \emph{causally-consistent}
  reversible by applying the approach in \cite{MelgrattiM0PU20}.
\item We finally show that the reversible unravel nets derived by our encoding
  are an interpretation of RCCS terms.
\end{enumerate}
An interesting aspect of the proposed encoding is that it highlights that all
the information needed for reversing an RCCS process is already \emph{encoded}
in the structure of the net corresponding to the original CCS process, i.e.,
RCCS memories are represented by the structure of the net.
Concretely, if an RCCS process $R$ is a reachable process from a  CCS process $P$ with empty memory,
then the encoding of $R$ is retrieved from the encoding of $P$, what changes
is the position of the markings.
Consider the CCS process $P = a.\zero$ that executes $a$ and then terminates.
It can be encoded as the Petri net on the left in \Cref{fig:intro} (the usage
of the apparently redundant coloured places in the postset of $a$ will be made clearer
in \Cref{sec:nets}).

\begin{figure}[t]
  \centering
  \scalebox{0.8}{\begin{tikzpicture}[scale=.8]
\tikzstyle{inhibitorred}=[o-, draw=red,thick]
\tikzstyle{unravel}=[->, draw=blue,thick]
\tikzstyle{pre}=[<-,thick]
\tikzstyle{post}=[->,thick]
\tikzstyle{readblue}=[-, draw=blue,thick]
\tikzstyle{transition}=[rectangle, draw=black,thick,minimum size=5mm]
\tikzstyle{revtransition}=[rectangle, draw=red!80,fill=red!30,thick,minimum size=5mm]
\tikzstyle{place}=[circle, draw=black,thick,minimum size=5mm]
\tikzstyle{placeblu}=[circle, draw=blue!80,fill=blue!20,thick,minimum size=5mm]

\node[place,tokens=1] (p1) at (0,0)  {};
\node[placeblu] (p2) at (3,-1) {};
\node[place,] (p3) at (3,1) {};

\node[transition] (b) at (1.5,0) {$a$}
edge[pre] (p1)
edge[post](p2)
edge[post](p3);

\end{tikzpicture}}\hspace{1.3cm}
  \scalebox{0.8}{\begin{tikzpicture}[scale=.8]
\tikzstyle{inhibitorred}=[o-, draw=red,thick]
\tikzstyle{unravel}=[->, draw=blue,thick]
\tikzstyle{pre}=[<-,thick]
\tikzstyle{post}=[->,thick]
\tikzstyle{readblue}=[-, draw=blue,thick]
\tikzstyle{transition}=[rectangle, draw=black,thick,minimum size=5mm]
\tikzstyle{revtransition}=[rectangle, draw=red!80,fill=red!30,thick,minimum size=5mm]
\tikzstyle{place}=[circle, draw=black,thick,minimum size=5mm]
\tikzstyle{placeblu}=[circle, draw=blue!80,fill=blue!20,thick,minimum size=5mm]

\node[place,tokens=1] (p1) at (0,0)  {};
\node[placeblu] (p2) at (3,-1) {};
\node[place] (p3) at (3,1) {};

\node[transition] (b) at (1.5,0) {$a$}
edge[pre] (p1)
edge[post](p2)
edge[post](p3);
\node[revtransition] (br) at (1.5,1) {$\un{a}$}
edge[post] (p1)
edge[pre](p2)
edge[pre](p3);

\end{tikzpicture}}\hspace{1.3cm}
  \scalebox{0.8}{\begin{tikzpicture}[scale=.8]
\tikzstyle{inhibitorred}=[o-, draw=red,thick]
\tikzstyle{unravel}=[->, draw=blue,thick]
\tikzstyle{pre}=[<-,thick]
\tikzstyle{post}=[->,thick]
\tikzstyle{readblue}=[-, draw=blue,thick]
\tikzstyle{transition}=[rectangle, draw=black,thick,minimum size=5mm]
\tikzstyle{revtransition}=[rectangle, draw=red!80,fill=red!30,thick,minimum size=5mm]
\tikzstyle{place}=[circle, draw=black,thick,minimum size=5mm]
\tikzstyle{placeblu}=[circle, draw=blue!80,fill=blue!20,thick,minimum size=5mm]

\node[place] (p1) at (0,0)  {};
\node[placeblu,tokens=1] (p2) at (3,-1) {};
\node[place,tokens=1] (p3) at (3,1) {};

\node[transition] (b) at (1.5,0) {$a$}
edge[pre] (p1)
edge[post](p2)
edge[post](p3);
\node[revtransition] (br) at (1.5,1) {$\un{a}$}
edge[post] (p1)
edge[pre](p2)
edge[pre](p3);

\end{tikzpicture}}
  \caption{Encoding of $R = \emp\proc a.\zero$}\label{fig:intro}
\end{figure}

The reversible version of $P$ is $R = \emp\proc a.\zero$, where $\emp$ denotes
an initially empty memory.
According to RCCS semantics, $R$ evolves to $R' = \mem{*}{a}{\zero}\emp\proc
\zero$ by executing $a$. The memory $\mem{*}{a}{\zero}\emp$ in $R'$ indicates
that it can go back to the initial process $R$ by undoing $a$.
Note that the net corresponding to $P$ (on the left) contains all the
necessary information to reverse the action $a$; intuitively, the action $a$
can be undone by firing it in the opposite direction (i.e., by consuming
tokens from the postset and producing them in its preset), or equivalently, by
executing a reversing transition $\un a$ as depicted in the net shown in the
middle of~\Cref{fig:intro}.
Furthermore, it is important to highlight that the net on the right of
\Cref{fig:intro} corresponds to the derivative $R'$.
The coloured place carries the information stored in the memory of the derivative $R'$.
Consequently, the encoding of a CCS term as a net already encompasses all the
information required for its reversal, which stands in contrast to the
additional memories needed in the case of RCCS.
This observation provides a straightforward and nearly immediate true
concurrent representation of RCCS processes, effectively capturing their
reversible behaviour.

\paragraph{Organization of the paper.} The paper is structured as follows:
after establishing essential notation, we provide a brief overview of CCS and
RCCS in \Cref{sec:sys}.
Next, in \Cref{sec:nets}, we present a concise summary of Petri nets and
introduce the concept of \emph{unravel} nets, followed by their reversible
counterpart. The encoding of CCS into unravel nets and the mapping of RCCS
terms into reversible unravel nets, along with correspondence results, are described in \Cref{sec:coding}.
In the final section, we draw insightful conclusions and discuss potential
avenues for future developments.
Additionally, we present a practical implementation of the encoding and
simulation of the execution in Haskell in \Cref{sec:impl}.
A preliminary version of this work has been published
as~\cite{MelgrattiMP21b}. In this version we have extended the scope and
applicability of the proposed approach. We move from finite processes to
infinite ones (i.e., recursive processes) by considering terms defined
coinductively. Secondly, in this version we provide full and rigorous proofs
of the key results. Finally, we provide a Haskell implementation of the
encoding that allows for the simulation of the execution of encoded CCS
processes. This practical implementation further exemplifies the feasibility
and effectiveness of the approach.


\subsection*{Preliminaries}\label{sec:prelimin}
We recall some notation that we will use in the paper.
We denote the set of natural numbers as $\mathbb{N}$.
Let $A$ be a set,  a \emph{multiset} of $A$ is defined as a function $m : A
\rightarrow \mathbb{N}$.
The set of multisets of $A$ is denoted by $\muset{A}$.
We assume the usual operations on multisets, such as union $+$ and difference
$-$.
For multisets $m, m' \in \muset{A}$, we write $m \subseteq m'$ to indicate
that $m(a) \leq m'(a)$ for all $a \in A$.
Additionally,
we define $\flt{m}$ as the multiset
where $\flt{m}(a) = 1$ if $m(a) > 0$, and $\flt{m}(a) = 0$ otherwise.
When a multiset $m$ of $A$ is a set, i.e., $m = \flt{m}$, we write $a \in m$
to denote that $m(a) \neq 0$. In this case, we often confuse the multiset $m$
with the set $\setcomp{a\in A}{m(a) \neq 0}$ or a subset $X\subseteq A$ with
the multiset $X(a) = 1$ if $a\in A$ and $X(a) = 0$ otherwise.
We also employ standard set operations such as $\cap$, $\cup$, or $\setminus$,
and, with a slight abuse of notation, write $\emptyset$ for the multiset $m$
such that $\flt{m} = \emptyset$.

Given a relation $\mathcal{R}$, we indicate with $\mathcal{R}^{*}$ its
reflexive and transitive closure.

\section{CCS and reversible CCS}\label{sec:sys}

Let $\mathcal{A}$ be a set of actions, denoted as $a, b, c, \ldots$, and let
$\co{\mathcal{A}} = \{ \co{a} \st a\in \mathcal{A}\}$ be the set of their
corresponding co-actions. The set containing all possible actions is denoted
by $\mathtt{Act}=\mathcal{A}\cup \co{\mathcal{A}}$.
We use $\alpha$ and $\beta$ to represent elements from
$\mathtt{Act}_{\tau} = \mathtt{Act} \cup \{\tau\}$, where $\tau$ is a symbol not
present in $\mathtt{Act}$, i.e., $\tau \notin \mathtt{Act}$, and denotes a
\textit{silent} action. We assume that for each $\alpha \in \mathtt{Act}$ we have that $\bar{\bar{\alpha}} = \alpha$.

\begin{figure}[th]
  \[
    \begin{array}{llll}
      (\text{Actions}) \quad
      & \alpha
      & \sdef a \ou \co{a} \ou \tau
      \\[10pt]
      (\text{CCS Processes}) \quad
      & P,Q
      &\sdefco  \sum_{i \in I} \alpha_i.P_i \ou (P\para Q) \ou P\backslash a
    \end{array}
  \]
  \caption{CCS Syntax}
  \label{fig:ccs_syn}
\end{figure}
\noindent 
The syntax of CCS is presented in Figure~\ref{fig:ccs_syn}.
A prefix (or action) in CCS can take one of three forms: an input $a$, an
output $\co{a}$, or the silent action $\tau$.
A term of the form $\sum_{i \in I} \alpha_i.P_i$ represents a process that
non-deterministically starts by selecting and performing some action
$\alpha_i$ and then continues as $P_i$.
We use $\nil$, the idle process, when $I=\emptyset$ in place of
$\sum_{i \in I} \alpha_i.P_i$. Similarly, we use $\alpha_i.P$ for a unitary
sum where $I$ is the singleton $\{i\}$.
The term $P\para Q$ represents the parallel composition of processes $P$ and
$Q$.
An action $a$ can be restricted to be visible only inside process $P$, denoted
as $P \backslash a$. Restriction is the only binder in CCS, where $a$ is bound
in $P \backslash a$.
We addressed the representation of infinite processes by adopting an approach
initiated by~\cite{castagna2009theory}. Instead of fixing a syntactic
representation of recursion, we simplified the treatment by employing infinite
regular trees. Throughout this paper, in Figure~\ref{fig:ccs_syn} and beyond,
we use the symbol $\sdefco$ to indicate that the productions should be
interpreted \emph{coinductively}. As a result, the set of processes is the
greatest fixed point of the (monotonic) functor over sets defined by the
grammar above~\cite{BarbaneraDd22}. Consequently, a process is a potentially
infinite, \emph{regular} term coinductively generated by the grammar in
Figure~\ref{fig:ccs_syn}. A term is considered regular if it consists of
finitely many \emph{distinct} subterms. The language generated by the
coinductive grammar is thus finitely representable either using the so-called
$\mu$ notation~\cite{Pierce02} or as solutions of finite sets of
equations~\cite{infiniteTree}. For a more comprehensive treatment, interested
readers are referred to~\cite{infiniteTree}.

We represent the set of all CCS processes as $\procs$.
We denote the set of names of a process $P$ as $\names(P)$, and we use
$\fn(P)$ and $\bn(P)$ to represent the sets of free and bound names in $P$,
respectively. (These functions can be straightforwardly defined by
coinduction.)

\begin{defi}[CCS Semantics]
  The operational semantics of CCS is defined as the LTS
  $(\procs,\mathtt{Act}_{\tau}, \fw )$ where the transition relation $\fw$ is
  the smallest relation induced by the rules in Figure~\ref{fig:ccs_sem}.
\end{defi}

\noindent 
Let us provide some comments on the rules presented in
Figure~\ref{fig:ccs_sem}. The \textsc{act} rule indicates that a
non-deterministic choice proceeds by executing one of its prefixes $\alpha_z$
and transitions to the corresponding continuation $P_z$. The \textsc{par-l}
and \textsc{par-r} rules allow the left and right processes of a parallel
composition to independently execute an action while the other remains
unchanged. The \textsc{syn} rule regulates synchronisation, allowing two
processes in parallel to perform a handshake. Lastly, the \textsc{hide} rule
restricts a certain action from being further propagated.

\begin{figure}[t] 
  \begin{mathpar}
    \inferrule*[Right=\raisebox{.1cm}{{\scriptsize(\textsc{act})}}]
    {z\in I}
    { \sum_{i\in I}\alpha_i.P_i\lts{\alpha_z} P_z }
    \and
    \inferrule*[Right={\scriptsize(\textsc{par-l})}]
    {P\lts{\alpha} P' }
    {\vphantom{ \sum_{i\in I}\alpha_i.P_i\lts{\alpha_z} P_z }P\para Q \lts{\alpha} P'\para Q}
    \and
    \inferrule*[Right={\scriptsize(\textsc{par-r})}]
    {Q\lts{\alpha} Q' }
    {\vphantom{ \sum_{i\in I}\alpha_i.P_i\lts{\alpha_z} P_z }P\para Q \lts{\alpha} P\para Q'}
    \\
    \inferrule*[Right={\scriptsize(\textsc{syn})}]
    {P\lts{\alpha} P' \and Q\lts{\bar{\alpha}} Q'}
    {P\para Q \lts{\tau} P'\para Q'}
    \and
    \inferrule*[Right={\scriptsize(\textsc{r-res})}]
    {P\lts{\alpha} P' \and \alpha\notin \{a,\bar{a}\}}
    {P\backslash a\lts{\alpha} P'\backslash a }
  \end{mathpar}
  \caption{CCS semantics}
  \label{fig:ccs_sem}
\end{figure}

\subsection{Reversible CCS}
Reversible CCS (RCCS)~\cite{rccs,rccsnew} is a reversible variant of CCS. In
RCCS, processes are equipped with a \textit{memory} that stores information
about their past actions.
The syntax of RCCS, shown in Figure~\ref{fig:rccs_syntax}, includes the same
constructs as the original CCS formulation, but with the addition of
reversible processes.
A reversible process in RCCS can take one of the following forms: a
\emph{monitored} process $m \proc P$ where $m$ represents the memory, and
$P$ is a CCS process; the parallel composition $R\para S$ of the reversible
processes $R$ and $S$; and the restriction $R\hide a$, where the action $a$ is
restricted to the process $R$.
A \textit{memory} is essentially a stack of events that encodes the history of
actions previously performed by a process. The left-most element in the memory
corresponds to the very last action executed by the monitored process.
Memories in RCCS can contain three different kinds of events\footnote{In this
  paper, we adopt the original RCCS semantics with partial synchronisation.
  Later versions, such as \cite{rccsnew}, employ communication keys to
  uniquely identify actions.}: \emph{partial} synchronisations
$\event{*}{\alpha}{Q}$, \emph{full} synchronisations $\event{m}{\alpha}{Q}$,
and memory \emph{splits} $\entry{0}$ and $\entry{1}$.
In a synchronisation, whether partial or full, the action $\alpha$ and the
process $Q$ serve specific purposes in recording the selected action $\alpha$
of a choice and the discarded branches $Q$. The technical distinction between
partial and full synchronisation will become evident when describing the
semantics of RCCS.
Events $\entry{0}$ and $\entry{1}$ represent the splitting of a monitored process into
two parallel ones, respectively the left one ($\entry{0}$) and the right one ($\entry{1}$).
The empty memory is represented by $\emp$.
Let us note that in RCCS, memories also serve as unique process identifiers,
 and this will be handy when undoing a full synchronisation.

\begin{figure}[t]
  \[
    \begin{array}{llll}
      (\text{CCS Processes}) \quad
      &P,Q
      &\sdefco \sum_{i \in I} \alpha_i.P_i
        \ou (P\para Q)
        \ou P\backslash a
      \\[10pt]
      (\text{RCCS Processes})\quad
      &R,S
      &\sdef m \proc P \ou (R\para S) \ou R\backslash a
      \\[10pt]
      (\text{Memories})\quad
      &m
      &\sdef\mem{*}{\alpha}{Q}{m}
        \ou \mem{m}{\alpha}{Q}{m'}
        \ou \lpar{m} \ou \rpar{m} \ou \emp
    \end{array}
  \]
  \caption{RCCS syntax}
  \label{fig:rccs_syntax}
\end{figure}

\begin{figure}[t] 
    \centering
    $\begin{array}{c}
    \inferrule*[Left=\raisebox{.1cm}{\scriptsize(\textsc{r-act})}]
    { }
    {m\proc \sum_{i\in I}\alpha_i.P_i\ltsk{m}{\alpha^{}_z} \mem{*}{\alpha^{z}_z}{\sum_{i\in I\setminus\{z\}}\alpha_i.P_i}{m}\proc P_z }
    \\[10pt]

    \inferrule*[Right=\raisebox{.1cm}{\scriptsize(\textsc{r-act$\rev$})}]
    { }
    {\mem{*}{\alpha^{z}_z}{\sum_{i\in I\setminus\{z\}}\alpha_i.P_i}{m}\proc P_z \rltsk{m}{\alpha_z}m\proc \sum_{i\in I}\alpha_i.P_i}
    \end{array}$
    
    \begin{align*}
    \inferrule*[Left={\scriptsize(\textsc{r-par-l})}]
    {R\ltsk{m}{\alpha} R'}
    {\vphantom{R\para S \rltsk{m}{\alpha} R'\para S}R\para S \ltsk{m}{\alpha} R'\para S}\:\ 
    &
    \ \:\inferrule*[Right={\scriptsize(\textsc{r-par-l}$\rev$)}]
    {R\rltsk{m}{\alpha} R'}
    {R\para S \rltsk{m}{\alpha} R'\para S}
    \\[8pt]
    \inferrule*[Left={\scriptsize(\textsc{r-par-r})}]
    {S\ltsk{m}{\alpha} S'}
    {\vphantom{R\para S \rltsk{m}{\alpha} R\para S'}R\para S \ltsk{m}{\alpha} R\para S'}\:\ 
    &
    \ \:\inferrule*[Right={\scriptsize(\textsc{r-par-r}$\rev$)}]
    {S\rltsk{m}{\alpha} S'}
    {R\para S \rltsk{m}{\alpha} R\para S'}
    \\[8pt]
    \inferrule*[Left={\scriptsize(\textsc{r-syn})}]
    {R\ltsk{m_1}{\alpha} R' \and S\ltsk{m_2}{\bar{\alpha}} S'}
    {\vphantom{R_{m_2@m_1}\para S_{m_1@m_2} \rltsk{m_1,m_2}{\tau} R'\para S'}R\para S \ltsk{m_1,m_2}{\tau} R'_{m_2@m_1}\para S'_{m_1@m_2}}\:\ 
    &
    \ \:\inferrule*[Right={\scriptsize(\textsc{r-syn}$\rev$)}]
    {R\rltsk{m_1}{\alpha} R' \and S\rltsk{m_2}{\bar{\alpha}} S'}
    {R_{m_2@m_1}\para S_{m_1@m_2} \rltsk{m_1,m_2}{\tau} R'\para S'}
    \\[8pt]
    \inferrule*[Left={\scriptsize(\textsc{r-res})}]
    {R\ltsk{m}{\alpha} R' \and \alpha\notin \{a,\bar{a}\}}
    {\vphantom{R\backslash a \rltsk{m}{\alpha} R'\backslash a}R\backslash a \ltsk{m}{\alpha} R'\backslash a}\:\ 
    &
    \ \:\inferrule*[Right={\scriptsize(\textsc{r-res}$\rev$)}]
    {R \rltsk{m}{\alpha} R' \and \alpha\notin \{a,\bar{a}\}}
    {R\backslash a \rltsk{m}{\alpha} R'\backslash a}
    \\[8pt]
    \inferrule*[Left={\scriptsize(\textsc{r-equiv})}]
    {R\con R' \quad R' \ltsk{m}{\alpha}S' \quad  S'\con S}
    {\vphantom{R\rltsk{m}{\alpha} S}R\ltsk{m}{\alpha} S}\ \:
    &
    \ \:\inferrule*[Right={\scriptsize(\textsc{r-equiv}$\rev$)}]
    {R\con R' \quad R' \rltsk{m}{\alpha} S' \quad S'\con S}
    {R\rltsk{m}{\alpha} S}
    \end{align*}
    \caption{RCCS semantics}
    \label{fig:rccs_sem}
\end{figure}

We define the following sets: the set $\procr$ of all RCCS processes, the set
$\mathcal{M}$ of all possible memories, and
$\hat{\mathcal{M}}=\mathcal{M}\cup\mathcal{M}^2$, which includes individual as
well as pairs of memories. We let $\hat{m}$ to range over the set
$\hat{\mathcal{M}}$.

As for CCS, the only binder in RCCS is restriction, which applies at the level
of both CCS and RCCS processes. Consequently, we extend the functions
$\names$, $\fn$, and $\bn$ to RCCS processes and memories accordingly.

\begin{defi}
  The operational semantics of RCCS is defined as a pair of LTSs sharing the
  same set of states and labels: a forward LTS
  $(\procr,\hat{\mathcal{M}} \times \mathtt{Act}_{\tau} , \fw )$ and a
  backward LTS $(\procr,\hat{\mathcal{M}} \times \mathtt{Act}_{\tau}, \bk )$.
  Elements of the set $\hat{\mathcal{M}} \times \mathtt{Act}_{\tau}$ will be denoted as
  $m:\alpha$.
  The transition relations $\fw$ and $\bk$ are the smallest relations induced
  by the rules in Figure~\ref{fig:rccs_sem} (left and right columns,
  respectively). Both relations make use of the structural congruence relation
  $\equiv$, which is the smallest congruence on RCCS processes containing the
  rules shown in Figure~\ref{fig:struct}.
  We define ${\rccst{\, } = {\fw \cup \bk}}$.
\end{defi}

Let us provide some comments on the forward rules in Figure \ref{fig:rccs_sem}
(left column).
Rule \textsc{r-act} allows a monitored process to perform a forward action
$\alpha_z$. Notably, the label of this transition pairs the executed action
$\alpha_z$ with the memory $m$ of the process.
At this point, we are uncertain whether the performed action will synchronise
with the context or not. Consequently, a partial synchronisation event of the
form $\event{*}{\alpha^z_z}{\sum_{i\in I\setminus{z}}\alpha_i.P_i}$ is added
on top of the memory.
The `*' in the partial synchronisation event will be replaced by a memory,
let's say $m_1$, if the process eventually synchronises with another process
monitored by $m_1$.
Additionally, it is essential to note that the discarded process 
$\sum_{i\in I\setminus{z}}\alpha_i.P_i$ is
recorded in the memory.
Moreover, along with the prefix, we store its position `$z$' within the sum.
While this piece of information may be redundant for RCCS itself and was not
present in the original semantics, it becomes useful when encoding an RCCS
process into a net and when proving operational correspondence. This
additional information enables a more straightforward representation of RCCS
processes in a net-based setting and supports the validation of operational
correspondence between the LTS and the net semantics.
Importantly, it is worth mentioning that this straightforward modification
does not alter the original semantics of RCCS.

Rules \textsc{r-par-l} and \textsc{r-par-r} allows for the independent
execution of an action in different components of a parallel composition.
Rule \textsc{r-syn} allows two parallel processes to synchronise. For
synchronisation to occur, the action $\alpha$ in one process must match the
co-action $\co{\alpha}$ in the other process.
Once this condition is met, the two partial synchronisations are updated to
two full synchronisations using the operator `@'.

\begin{defi}
Let $R$ be a monitored
process, and let $m_1$ and $m_2$ be two memories.
$R_{m_2@m_1}$ represents the
process obtained from $R$ by substituting all occurrences of
$\mem{*} {\alpha}{ Q} {m_1}$ with $\mem{m_2} {\alpha}{ Q} {m_1}$.
\end{defi}
\noindent 
Rule \textsc{r-res} propagates actions through restriction, provided that the
action is not on the restricted name.
\noindent 
Rule \textsc{r-equiv} allows one to exploit the structural congruence defined
in~\Cref{fig:struct}. The structural rule \textsc{split} enables a monitored
process with a top-level parallel composition to split into left and right
branches, resulting in the duplication of the memory. The structural rule
\textsc{res} permits pushing restrictions outside monitored processes. Lastly,
the structural rule $\alpha$ allows one to take advantage of
$\alpha$-conversion, denoted by $=_{\alpha}$.

\begin{figure}[t]
  \begin{center}
    \[
      \begin{array}{l@{\hspace{.5cm}}l}
        \textsc{(split)}
        & m\proc (P\para Q) \con \lpar{m}\proc P\para \rpar{m}\proc Q
        \\[10pt]
        \textsc{(res)}
        & m\proc P\backslash a \con (m\proc P)\backslash a
          \qquad \text{if }a\notin \fn(m)
        \\[10pt]
        \textsc{($\alpha$)}
        & R\con S \qquad\text{ if } R\aeq S
      \end{array}
    \]
  \end{center}
  \caption{RCCS Structural laws}
  \label{fig:struct}
\end{figure}

Backward rules are reported in the right column of the Figure
\ref{fig:rccs_sem}. As one can see, for each forward rule there exists a
symmetrical backward one. Rule \textsc{r-act$\rev$} allows a monitored process
to undo its last action, which coincides with the event on top of the memory
stack. As we can see, all the information is stored in the last performed
event, hence the rule pops out the last event on the memory, and restores back
the prefix corresponding to the event and the plus context. Rules
\textsc{r-par-l$\rev$} and \textsc{r-par-r$\rev$} allow for the independent
undoing of an action in different components of a parallel composition. Rule
\textsc{r-syn$\rev$} allows for a de-synchronisation: that is, two parallel
components which participated to a synchronisation, say, with labels $\alpha$
and $\co{\alpha}$ can undo this synchronisation. Let us stress out that two
processes, say $R$ and $S$ can undo a synchronisation along memories $m_1$ and
$m_2$ only if they are in the form $R_{m_2@m_1}$ and $S_{m_1@m_2}$. Rules
\textsc{r-res$\rev$} and \textsc{r-equiv$\rev$} acts like their forward
counterparts.

\begin{defi}
\label{lbl:initial_proc}
  An RCCS process of the form $\emp \proc P$ is referred to as \emph{initial}.
  Any process $R$ derived from an initial process using the rules in
  Figure \ref{fig:rccs_sem} is called \emph{coherent}.
\end{defi}

\begin{exa}
  Let $P = a.(b\parallel c) \parallel (\co{a} \parallel d)$. Via an
  application of the \textsc{Split}  rule we obtain the following process
  \[    \emp\proc P \equiv \,
      \big (\lpar{\emp}\proc a.(b\parallel c) \big ) \parallel \big(\rpar{\emp}\proc (\co{a} \parallel d)\big ) = R_1
 \]
 and we can further apply \textsc{Split} rule on the second monitored process of $R_1$ as follows:
\[         R_1 
    \equiv \, \lpar{\emp}\proc a.(b\parallel c) \parallel \lpar{\rpar{\emp}}\proc \co{a}
      \parallel \rpar{\rpar{\emp}}  \proc d = R_2
 \]
  Now, in the process $R_2$ there are two monitored processes which can
  synchronise on $a$. That is
\[    R_2\ltsk{m_1,m_2}{\tau}
     \event{m_1}{a}{\nil} \cdot \lpar{\emp}\proc (b \parallel c)  \parallel
      \big(\event{m_2}{\co a}{\nil} \cdot \lpar{\rpar{\emp}}\proc \nil \big) \parallel
      \big(\rpar{\rpar{\emp}}\proc d\big) = R_3
     \]
     and by applying the \textsc{Split} rule on the left-most monitored process of $R_3$ we obtain 
  \begin{align*}   
    R_3\equiv \,
    & \big(\lpar{\event{m_1}{a}{\nil} \cdot \lpar{\emp}} \proc b\big) \parallel
      \big(\rpar{\event{m_1}{a}{\nil} \cdot \lpar{\emp}} \proc c\big) \parallel
    \\
    &\big( \event{m_2}{\co a}{\nil} \cdot \lpar{\rpar{\emp}}\proc \nil\big) \parallel
      \big(\rpar{\rpar{\emp}}\proc d \big)
  \end{align*}
  where $m_1 = \lpar{\emp}$ and $m_2 = \rpar{\emp}$.
\end{exa}

An important property of a fully reversible calculus is the so called
  Loop Lemma, stating that any action can be undone. Formally:
\begin{lem}[Loop Lemma \cite{rccs}]
   Let $R$ be a coherent process. For any forward transition
    $R \ltsk{\hat{m}}{\alpha} S$ there exists a backward transition
    $S \rltsk{\hat{m}}{\alpha} R$, and conversely.
\end{lem}

\begin{cor} \label{cor:loop}
Let $R$ be a coherent process. If $R \rccst{\,}^*R_1$ then
$R_1\rccst{\,}^*R$.
\end{cor}
\noindent 
RCCS is shown to be causal consistent, that is any step can be undone provided that
its consequences are undone beforehand. A consequence of causal consistent reversibility,
it that any process reached by mixing computations (e.g., forward and backward transitions)
can be reached by only forward computations. That is:
\begin{pty}\label{prp:fw_only}
For any initial procces $P$, if $\emp\proc P \rccst{}^* R$ then
$\emp\proc P \rightarrow^*R$.
\end{pty}
\noindent 
The  notion of context below  will be useful in the following sections.
\begin{defi}
  RCCS process context $C$ and active contexts $E$ are reversible processes
  with a hole ``$\circ$'', defined by the following grammar:
  \begin{align*}
    &C\ \sdefco \circ\ou m\proc C \ou \alpha.C\ou 
    \sum_{i \in I} P^{C}_i
 \ou C\ou (P \parallel C) \ou (C \parallel P) \ou C\hide{A} \\
    &E\ \sdefco \circ\ou  (R \parallel E) \ou (E \parallel R) \ou E \hide{A}
  \end{align*}
where $P^C$ can be either $C$ or $P$.
\end{defi}

\section{Petri nets, Unravel Nets and Reversible Unravel Nets}\label{sec:nets}

\subsection{Petri nets}
We provide a brief overview of Petri nets, along with some related auxiliary
notions.

\begin{defi}
  A \emph{Petri net} is a tuple $N = \macronet{}{}$, where $S$ is a set of
  {\em places}, $T$ is a set of {\em transitions} (with
  $S \cap T = \emptyset$), $F \subseteq (S\times T)\cup (T\times S)$ is the
  \emph{flow} relation, and $\marko{m}\in \muset{S}$ is the \emph{initial
    marking}.
\end{defi}
\noindent 
Petri nets are conventionally represented with transitions depicted as boxes,
places as circles, and the flow relation indicated by directed arcs.
The marking $\marko{m}$, i.e., the state of the net, is depicted by drawing
in the place $s$ a number $\marko{m}(s)$ of `$\bullet$' symbols, also called \emph{tokens}.

Given a net $N = \macronet{}{}$ and $x\in S\cup T$, we define the following
multisets: $\pre{x} = \setcomp{y}{(y,x)\in F}$ and
$\post{x} = \setcomp{y}{(x,y)\in F}$. If $x$ is a place then $\pre{x}$ and
$\post{x}$ are (multisets) of transitions; analogously, if $x\in T$ then
$\pre{x}\in\muset{S}$ and $\post{x} \in \muset{S}$.
The sets $\pre{x}$ and
$\post{x}$ are respectively called the \emph{pre} and \emph{postset} of $x$.
A transition $t\in T$ is enabled at a marking $\marko{m}\in \muset{S}$,
denoted by $\marko{m}\trans{t}$, whenever $\pre{t} \subseteq \marko{m}$.
A transition $t$ enabled at a marking $\marko{m}$ can \emph{fire} and its
firing produces the marking $\marko{m}' = \marko{m} - \pre{t} + \post{t}$. The
firing of $t$ at a marking $\marko{m}$ producing $\marko{m}'$ is denoted by
$\marko{m}\trans{t}\marko{m}'$.
We assume that each transition $t$ of a net $N$ is such that
$\pre{t}\neq\emptyset$, meaning that no transition may fire
\emph{spontaneously}.
Given a generic marking $\marko{m}$ (not necessarily the initial one), the
\emph{firing sequence} ({shortened as} \fs) of
$N = \langle S, T, F, \marko{m}\rangle$ starting at $\marko{m}_0$ is defined
as:

\begin{itemize}
\item $\marko{m}_0$ is a firing sequence (of length 0), and
\item if $\marko{m}_0\trans{t_1}\marko{m}_1$ $\cdots$
  $\marko{m}_{n-1}\trans{t_n}\marko{m}_n$ is a firing sequence and
  $\marko{m}_n\trans{t}\marko{m}'$, then also
  $\marko{m}_0\trans{t_1}\marko{m}_1$ $\cdots$
  $\marko{m}_{n-1}\trans{t_n}\marko{m}_n\trans{t}\marko{m}'$ is a firing
  sequence.
\end{itemize}
\noindent 
The set of firing sequences of a net $N = \langle S, T, F, \marko{m}\rangle$
starting at a marking $\marko{m}$ is denoted by $\firseq{N}{\marko{m}}$
and it is ranged over by $\sigma$. Given a \fs\
$\sigma = \marko{m}_0\trans{t_1}\sigma'\trans{t_n}\marko{m}_n$,
$\start{\sigma}$ is the marking $\marko{m}_0$, $\lead{\sigma}$ is the marking
$\marko{m}_n$ and $\remains{\sigma}$ is the \fs\
$\sigma'\trans{t_n}\marko{m}_n$.
Given a net $N = \langle S, T, F, \marko{m}\rangle$, a marking $\marko{m}'$
is \emph{reachable} iff there exists a \fs\
$\sigma \in \firseq{N}{\marko{m}}$ such that $\lead{\sigma}$ is $\marko{m}'$.
The set of reachable markings of $N$ is
$\reachMark{N} = \setcomp{\lead{\sigma}}{\sigma\in\firseq{N}{\marko{m}}}$.
Given a \fs\
$\sigma = \marko{m}\trans{t_1}\marko{m}_1\cdots
\marko{m}_{n-1}\trans{t_n}\marko{m}'$, we write
$X_{\sigma} = \sum_{i=1}^{n} t_i$ for the multiset of transitions associated
to \fs, which we call an \emph{execution} of the net and we write
\( \states{N} = \setcomp{X_{\sigma}\in
  \muset{T}}{\sigma\in\firseq{N}{\marko{\marko{m}}}} \) for the set of the
executions of $N$.
Observe that an execution simply says which transitions (and the relative
number of occurrences of them) has been executed, not their (partial)
ordering.
Given a \fs\
$\sigma = \marko{m}\trans{t_1}\marko{m}_1\cdots
\marko{m}_{n-1}\trans{t_n}\marko{m}_n\cdots$, with $\rho_{\sigma}$ we denote
the sequence $t_{1}t_{2}\cdots t_{n}\cdots$.

\begin{defi}
  A net $N = \macronet{}{}$ is said to be \emph{safe} if each marking
  $\marko{m}\in \reachMark{N}$ is such that $\marko{m} = \flt{\marko{m}}$.
\end{defi}
\noindent 
The notion of subnet will be handy in the following.
A subnet is obtained by restricting places and transitions, and
correspondingly the flow relation and the initial marking.
\begin{defi}
  \label{de:subnet-tran}
  Let $N = \macronet{}{}$ be a Petri net and let $T'\subseteq T$ be a subset
  of transitions and $S' = \pre{T'}\cup\post{T'}$. Then, the subnet generated
  by $T'$ $\subnet{N}{T'}{} = \macronet{}{'}$, where $F'$ is the restriction
  of $F$ to $S'$ and $T'$, and $\marko{\marko{m}}'$ is the multiset on $S'$
  obtained by $\marko{\marko{m}}$ restricting to the places in $S'$.
\end{defi}

\subsection{Unravel nets.}
To define \emph{unravel nets} we need the notion of \emph{causal net}, i.e., a net representing
how the various transitions are related and all of them can be executed in a firing sequence.

\begin{defi}\label{de:occ-net}
  A safe Petri net $N = \macronet{}{}$ is a \emph{causal net} (\ca{} for
  short) when $\forall s\in S$. $|\pre{s}| \leq 1$ and $|\post{s}| \leq 1$,
  $F^{\ast}$ is acyclic, $T\in \states{N}$, and
  $\forall s\in S\ \pre{s} = \emptyset\ \Rightarrow\ \marko{m}(s) = 1$.
\end{defi}
\noindent 
Requiring that $T\in \states{N}$ implies that all the transition can be
executed whereas $F^{\ast}$ acyclic means that dependencies among transitions
are settled.
Observe that causal net has no isolated and unmarked places as
$\forall s\in S\ \pre{s} = \emptyset\ \Rightarrow\ \marko{m}(s) = 1$.

\begin{defi}\label{de:unravel-net}
  An \emph{unravel net} (\unet{} for short) $N = \macronet{}{}$ is a safe net
  such that
  \begin{enumerate}
  \item for each execution $X\in\states{N}$ the subnet $\subnet{N}{X}{}$ is a
    \ca{}, and
  \item $\forall t, t'\in T$.
    $\pre{t} = \pre{t'}\ \land\ \post{t} = \post{t'}\ \Rightarrow\ t = t'$.
  \end{enumerate}
\end{defi}
\noindent 
Unravel nets describe the dependencies among transitions in the executions of
a concurrent and distributed device and are similar to \emph{flow
  nets}~\cite{BouFN90,BCIC94}.
Flow nets are safe nets in which, for every possible firing sequence, each
place can be marked only once.
The first condition in Definition~\ref{de:unravel-net} implies that the subnet
consisting of the transitions executed by the firing sequence is a causal net.
The second condition, which states that when two transitions have identical
pre- and postsets they are the same transition, serves the purpose of ruling out
the possibility of having two different transitions that are indistinguishable
because they consume and produce the same tokens (places).
Similar to flow nets, \unet{} also adhere to the rule that each place can be marked only once. However, unlike flow nets, the requirement that two transitions with the same preset and postset are the same transition (economy efficiency) stipulated by the second condition is integral to its definition.
Moreover, flow nets were initially introduced alongside flow event structure~\cite{BouFN90}, a concept which we do not consider in this paper.
Lastly, as the process algebra we consider cannot have terms with unguarded choices, the requirement that outgoing arcs are denumerable is always fulfilled and therefore we do not require it explicitly.
In an \unet{}, two transitions $t$ and $t'$ are conflicting if they never
appear together in an execution, \emph{i.e.}, $\forall X\in\states{N}$.
$\setenum{t, t'}\not\subseteq X$, as formally stated below.
Given a place $s$ of an unravel net, if $\pre{s}$ contains two or more
transitions, then they are in conflict.

\begin{prop}
  Let $N = \macronet{}{}$ be an \unet{} and $s\in S$ be a place such that
  $\card{\!\pre{s}} > 1$. Then
  $\forall t, t'\in \pre{s}.\ \forall X\in\states{N}$, if $t\in X$ and
  $t'\in X$ then $t = t'$.
\end{prop}
\begin{proof}
  Take a place $s\in S$ such that $\card{\!\pre{s}} > 1$ and take
  $t, t'\in \pre{s}$. Assume that there is an execution $X \in\states{N}$ such
  that $X$ contains both $t$ and $t'$. As $\subnet{N}{X}{}$ is a causal net we
  have that it is acyclic and therefore $t$ must be equal to $t'$ as both
  produce a token in $s$.
\end{proof}
\noindent 
It is worth noting that the classical notion of an \emph{occurrence
  net}~\cite{NPW:PNES,Win:ES} is, in fact, a specific type of \unet{}. In this
context, the conflict relation is \emph{inherited} throughout the transitive
closure of the flow relation and can be inferred directly from
the structure of the net itself.
A further evidence that unravel nets generalize occurrence nets is implied
also by the fact that flow nets generalize occurrence nets as
well~\cite{BouFN90}.

\begin{defi}
  An unravel net $N = \macronet{}{}$ is \emph{complete} whenever
  $\forall t\in T$. $\exists s_t\in S.$
  $\pre{s_t} = \setenum{t}\ \land\ \post{s_t} =\emptyset$, and
  $\card{\post{t}}>1$.
  We use $\keypl{T}$ to denote the subset of $S$ of such places and we call
  the places in $\keypl{T}$ \emph{key}-places. 
\end{defi}
\noindent 
We choose the term \emph{key-places} to denote the places within the
  set $\keypl{T}$, as they resemble the communication keys in
  CCKS~\cite{ccsk}. These keys serve as unique markers used to indicate
  partial or full synchronisation.
Thus, in a complete \unet{}, the execution of a transition $t$ is signalled by
the marked place $s_t$. Given an \unet{} $N$, it can be turned easily into a
complete one by adding for each transition the suitable place, without
changing the executions of the net, thus we consider complete \unet{s} only.
Completeness comes handy when defining the reversible counterpart of an \unet.
In a complete \unet{} $N = \macronet{}{}$, it is easy to see that $\card{\keypl{T}} = \card{T}$.

\begin{prop}
 Let $N = \macronet{}{}$ be a complete \unet{} and $\keypl{T}$ the key-places. Then,
 $\states{N} = \states{N'}$ where $N' = \langle S', T, F', \marko{m}\rangle$
 and $S' = S\setminus\keypl{T}$ and $F' = F \cap ((S'\times T) \cup (T\times S'))$.
\end{prop}
\noindent 
The key-places do not play any role in the firing of transitions in a \unet{}.

\begin{exa}
Consider the nets in \Cref{fig:unetvariants}. The net $N$ is an unravel net, that has the two maximal executions delineated by the following sequences: first $a$, followed by $c$, or first $b$, then $c$.
The net $N$ is not complete
due to the absence of key-places associated with transitions $a$, $b$, and $c$. Transitions $a$ and $b$ lack key-places because each place in their postsets possesses an outgoing arc. Transition $c$  also lacks a key-place because it has just one place in its postset.
The net $N'$ is derived from $N$ by  augmenting each transition with a key-place, rendering $N'$ a complete net. These key-places serve the purpose of recording executed transitions. Consequently, in executions such as first $a$ and then $c$, tokens are placed in the key-places corresponding to $a$ and $c$.
\end{exa}

In a complete net, certain key-places can be determined unambiguously, such as those associated with transitions $a$ and $b$ in \Cref{fig:unetcomp}. However, for others, the selection  is somewhat arbitrary, as one place is chosen among alternative options. This is exemplified by the key-place linked to transition $c$ in \Cref{fig:unetcomp}

\begin{figure}[thb]
\centering
\begin{subfigure}[b]{0.25\textwidth}
	\centering
      	\scalebox{0.8}{{\footnotesize
\begin{tikzpicture}[scale=.7]
\tikzstyle{inhibitorred}=[o-,draw=red,thick]
\tikzstyle{inhibitorblu}=[o-,draw=blue,thick]
\tikzstyle{pre}=[<-,thick]
\tikzstyle{post}=[->,thick]
\tikzstyle{prered}=[<-,thick,draw=red]
\tikzstyle{postred}=[->,thick,draw=red]
\tikzstyle{readblue}=[-, draw=blue,thick]
\tikzstyle{transition}=[rectangle, draw=black,thick,minimum size=5mm]
\tikzstyle{revtransition}=[rectangle, draw=red!80,fill=red!30,thick,minimum size=5mm]
\tikzstyle{place}=[circle, draw=black,thick,minimum size=5mm]
\tikzstyle{placeblu}=[circle, draw=blue!80,fill=blue!20,thick,minimum size=5mm]
\node[place,tokens=1] (p1) at (0,5)  {};
\node[place,tokens=1] (p2) at (1.5,5)  {};
\node[place,tokens=1] (p3) at (3,5)  {};
\node[place]          (p4) at (1.5,2.5) {};
\node[place]          (p5) at (1.5,0) {};
\node[transition] (a) at (.5,3.75)  {$a$}
edge[pre] (p1)
edge[pre] (p2)
edge[post](p4);
\node[transition] (b) at (2.5,3.75)  {$b$}
edge[post] (p4)
edge[pre](p3)
edge[pre](p2);
\node[transition] (c) at (1.5,1.25) {$c$}
edge[pre] (p4)
edge[post] (p5);
\end{tikzpicture}
}}
        \caption{$N$}
        \label{fig:unet}
\end{subfigure}
\hspace{.7cm}
\begin{subfigure}[b]{0.25\textwidth}
	\centering
      	\scalebox{0.8}{{\footnotesize
\begin{tikzpicture}[scale=.7]
\tikzstyle{inhibitorred}=[o-,draw=red,thick]
\tikzstyle{inhibitorblu}=[o-,draw=blue,thick]
\tikzstyle{pre}=[<-,thick]
\tikzstyle{post}=[->,thick]
\tikzstyle{prered}=[<-,thick,draw=red]
\tikzstyle{postred}=[->,thick,draw=red]
\tikzstyle{readblue}=[-, draw=blue,thick]
\tikzstyle{transition}=[rectangle, draw=black,thick,minimum size=5mm]
\tikzstyle{revtransition}=[rectangle, draw=red!80,fill=red!30,thick,minimum size=5mm]
\tikzstyle{place}=[circle, draw=black,thick,minimum size=5mm]
\tikzstyle{placeblu}=[circle, draw=blue!80,fill=blue!20,thick,minimum size=5mm]
\node[place,tokens=1] (p1) at (0,5)  {};
\node[place,tokens=1] (p2) at (1.5,5)  {};
\node[place,tokens=1] (p3) at (3,5)  {};
\node[place]          (p4) at (1.5,2.5) {};
\node[place]          (p5) at (1.5,0) {};
\node[placeblu]       (p6) at (0,2.5) {};
\node[placeblu]       (p7) at (3,2.5) {};
\node[placeblu]       (p8) at (0,0) {};
\node[transition] (a) at (.5,3.75)  {$a$}
edge[pre] (p1)
edge[pre] (p2)
edge[post] (p6)
edge[post](p4);
\node[transition] (b) at (2.5,3.75)  {$b$}
edge[post] (p4)
edge[post] (p7)
edge[pre](p3)
edge[pre](p2);
\node[transition] (c) at (1.5,1.25) {$c$}
edge[pre] (p4)
edge[post] (p8)
edge[post] (p5);
\end{tikzpicture}
}}
	\caption{$N'$}
         \label{fig:unetcomp}
\end{subfigure}
\hspace{.7cm}
\begin{subfigure}[b]{0.35\textwidth}
	\centering
      	\scalebox{0.8}{{\footnotesize
\begin{tikzpicture}[scale=.7]
\tikzstyle{inhibitorred}=[o-,draw=red,thick]
\tikzstyle{inhibitorblu}=[o-,draw=blue,thick]
\tikzstyle{pre}=[<-,thick]
\tikzstyle{post}=[->,thick]
\tikzstyle{prered}=[<-,thick,draw=red]
\tikzstyle{postred}=[->,thick,draw=red]
\tikzstyle{readblue}=[-, draw=blue,thick]
\tikzstyle{transition}=[rectangle, draw=black,thick,minimum size=5mm]
\tikzstyle{revtransition}=[rectangle, draw=red!80,fill=red!30,thick,minimum size=5mm]
\tikzstyle{place}=[circle, draw=black,thick,minimum size=5mm]
\tikzstyle{placeblu}=[circle, draw=blue!80,fill=blue!20,thick,minimum size=5mm]
\node[place] (p1) at (0,5)  {};
\node[place] (p2) at (1.5,5)  {};
\node[place,tokens=1] (p3) at (3,5)  {};
\node[place]          (p4) at (1.5,2.5) {};
\node[place,tokens=1]          (p5) at (1.5,0) {};
\node[placeblu,tokens=1]       (p6) at (0,2.5) {};
\node[placeblu]       (p7) at (3,2.5) {};
\node[placeblu,tokens=1]       (p8) at (0,0) {};
\node[transition] (a) at (.5,3.75)  {$a$}
edge[pre] (p1)
edge[pre] (p2)
edge[post] (p6)
edge[post](p4);
\node[transition] (b) at (2.5,3.75)  {$b$}
edge[post] (p4)
edge[post] (p7)
edge[pre](p3)
edge[pre](p2);
\node[transition] (c) at (1.5,1.25) {$c$}
edge[pre] (p4)
edge[post] (p8)
edge[post] (p5);
\end{tikzpicture}
}}
        \caption{$N'$}
        \label{fig:unetrev}
\end{subfigure}
\caption{An \unet{} $N$, its complete version $N'$ and the net $N'$ after the execution of $a$ and $b$ }\label{fig:unetvariants}
\end{figure}

\subsection{Reversible unravel nets.}
The definition of \emph{reversible unravel nets} builds upon that of the
\emph{reversible occurrence nets} of~\cite{MelgrattiM0PU20},
extending the notion just as unravel nets generalise upon occurrence nets.

\begin{defi}\label{def:rcn}
  A \emph{reversible unravel net} (\rcn{} for short) is a quintuple
  $N = \macrorevnet{}{}$ such that
  \begin{enumerate}
  \item $\Er\subseteq T$ and $\forall \er\in \Er$.
    $\exists!\ t\in T\setminus\Er$ such that $\pre{\er} = \post{t}$ and
    $\post{\er} = \pre{t}$, and
  \item $\subnet{N}{T\setminus \Er}{}$ is a complete unravel net and
    $\macronet{}{}$ is a safe one.
  \end{enumerate}
\end{defi}
\noindent 
The transitions in $\Er$ are the reversing ones; hence, we often say that a
reversible unravel net $N$ is \emph{reversible with respect to} $\Er$.
A reversing transition $\er$ is associated with a unique non-reversing
transition $t$ (condition 1) and its effects are intended to \emph{undo} $t$.
This fact ensures the existence of an injective mapping
$h : \Er \to T\setminus\Er$, which consequently implies that each reversible
transition is accompanied by precisely one corresponding reversing transition.
The second condition stipulates that when disregarding all reversing
transitions, the resulting subnet is indeed a complete unravel net and the net
itself is a safe net.

Along the lines of \cite{MelgrattiM0PU20}, we can prove that the set of
reachable markings of a reversible unravel net is not influenced by performing
a reversing transition.
\begin{prop}\label{pr:reachable-markings-of-rcn}
  Let $N = \macrorevnet{}{}$ be an \rcn. Then
  $\reachMark{N} = \reachMark{\subnet{N}{T\setminus\Er}{}}$.
\end{prop}
\begin{proof}
  Clearly $\reachMark{\subnet{N}{T\setminus\Er}{}}\subseteq \reachMark{N}$.
  For the other inclusion, we first observe that if $\marko{m}\trans{t}$ then
  $t\in T\setminus\Er$ as none of the transitions in $\Er$ is enabled at the
  initial marking. Consider now an \fs{} $\sigma\trans{\er}m$, with
    $u\in \Er$, and w.l.o.g. assume that all the transitions in $\sigma$
  belong to $T\setminus\Er$, i.e., $X_{\sigma}\subseteq T\setminus\Er$. We
  construct an \fs{} leading to $m$ which does not contain any transition in
  $\Er$. As $\sigma\trans{\er}$ we have that
  $\pre{\er}\subseteq \lead{\sigma}$ and this implies that the transition
  $h(\er)\in X_{\sigma}$. We can then write $\sigma$ as
  $\sigma'\trans{h(\er)}\sigma''$ and none of the transitions in $\sigma''$
  uses the tokens produced by $h(\er)$ as $\subnet{N}{X_{\sigma}}{}$ is a
  subnet of $\subnet{N}{T\setminus\Er}{}$, which is a complete \unet{}.
  Therefore we have that the transitions in the \fs{}
  $\lead{\sigma'}\trans{h(\er)}\sigma''$ can be rearranged in a \fs{}
  $\sigma'''\trans{h(\er)}\lead{\sigma}$. Observing that the effects of firing
  $\er$ at $\lead{\sigma}$ are producing the tokens in places $\pre{h(\er)}$
  we have that the \fs{} we are looking for is obtained executing the
  transitions in $\sigma'$ followed by the ones in $\sigma'''$ and the reached
  marking is precisely $\lead{\sigma}$. Hence also
  $\reachMark{N} \subseteq \reachMark{\subnet{N}{T\setminus\Er}{}}$ holds.
\end{proof}
\noindent 
A consequence of this fact is that each marking can be reached by using just
\emph{forward transition}.

Given an unravel net and a subset of transitions to be reversed, it is
straightforward to obtain a reversible unravel net.
\begin{prop}\label{pr:constructing-a-reversible-unravel-net}
  Let $N = \macronet{}{}$ be a complete unravel
  net
  and let $\Er \subseteq T$ be the set of transitions to be reversed. Define
  $\revunr N U = \macrorevnet{}{\,'}$ where $S = S\,'$,
  $\Er\,' = \Er\times\setenum{\re}$, $T\,' = (T\times\setenum{\fe}) \ \cup\ \Er\,'$,
  \[
    \begin{array}{ll}
      F\,' =
      & \setcomp{(s, (t,\fe))}{(s,t)\in F}\ \cup\ \setcomp{((t,\fe),s)}{(t,s)\in F}\ \cup\
      \\
      &\setcomp{(s, (t,\re))}{(t,s)\in F}\ \cup\ \setcomp{ ((t,\re),s)}{(s,t)\in F}
    \end{array}
  \]
  and $\marko{\marko{m}}' = \marko{\marko{m}}$. Then $\revunr N U$ is a
  reversible unravel net.
\end{prop}
\begin{proof}
  We check the conditions of \cref{def:rcn}. The first condition is satisfied
  as we observe that for each transition in $(t,\re)\in \Er\,'$, there exists a
  unique corresponding transition $(t,\fe)\in T\times\setenum{\fe}$; moreover,
  $\pre{(t,\re)} = \post{(t,\fe)}$ and $\post{(t,\re)} = \pre{(t,\fe)}$. The
  second one depends on the fact that $N$ is a complete \unet. Finally $N$ is,
  up to the renaming of transitions, equal to $\subnet{\revunr N U}{\Er'}{}$,
  which is a complete unravel net. Finally, $\revunr N U$ is trivially safe as
  $N$ is safe.
\end{proof}
\noindent 
The construction above simply adds as many events (transitions) as transitions
to be reversed in $\Er$. The preset of each added event is the postset of the
corresponding event to be reversed, and its postset is the preset of the event
to be reversed.
We write $\revnet{N}$ instead of $\revunr N T$ when $N = \macronet{}{}$, i.e.,
when every transition is reversible.
\noindent 
The reversible unravel net obtained by reversing
every transition is depicted in \Cref{fig:unetrev}.

\begin{figure}[!t]
\centering
\begin{subfigure}[b]{0.30\textwidth}
	\centering
      	\scalebox{0.8}{{\footnotesize
\begin{tikzpicture}[scale=.7]
\tikzstyle{inhibitorred}=[o-,draw=red,thick]
\tikzstyle{inhibitorblu}=[o-,draw=blue,thick]
\tikzstyle{pre}=[<-,thick]
\tikzstyle{post}=[->,thick]
\tikzstyle{prered}=[<-,thick,draw=red]
\tikzstyle{postred}=[->,thick,draw=red]
\tikzstyle{readblue}=[-, draw=blue,thick]
\tikzstyle{transition}=[rectangle, draw=black,thick,minimum size=5mm]
\tikzstyle{revtransition}=[rectangle, draw=red!80,fill=red!30,thick,minimum size=5mm]
\tikzstyle{place}=[circle, draw=black,thick,minimum size=5mm]
\tikzstyle{placeblu}=[circle, draw=blue!80,fill=blue!20,thick,minimum size=5mm]
\node[place,tokens=1] (p1) at (0,5)  {};
\node[place,tokens=1] (p2) at (1.5,5)  {};
\node[place,tokens=1] (p3) at (3,5)  {};
\node[place]          (p4) at (1.5,2.5) {};
\node[place]          (p5) at (.5,0) {};
\node[placeblu]       (p6) at (0,2.5) {};
\node[placeblu]       (p7) at (3,2.5) {};
\node[placeblu]       (p8) at (-1,0) {};
\node[place]          (p9) at (2.5,0) {};
\node[placeblu]       (p10) at (4,0) {};
\node[transition] (a) at (.5,3.75)  {$a$}
edge[pre] (p1)
edge[pre] (p2)
edge[post] (p6)
edge[post](p4);
\node[revtransition] (ar) at (-1,3.75)  {$\un{a}$}
edge[postred] (p1)
edge[postred] (p2)
edge[prered] (p6)
edge[prered](p4);
\node[transition] (b) at (2.5,3.75)  {$b$}
edge[post] (p4)
edge[post] (p7)
edge[pre](p3)
edge[pre](p2);
\node[transition] (c) at (0.5,1.25) {$c$}
edge[pre] (p4)
edge[post] (p8)
edge[post] (p5);
\node[revtransition] (cr) at (-1,1.25) {$\un{c}$}
edge[postred] (p4)
edge[prered] (p8)
edge[prered] (p5);
\node[transition] (d) at (2.5,1.25) {$d$}
edge[pre] (p4)
edge[post] (p9)
edge[post] (p10);
\end{tikzpicture}
}}
	\caption{$N$}
         \label{fig:unetcompuno}
\end{subfigure}
\hspace{.7cm}
\begin{subfigure}[b]{0.25\textwidth}
	\centering
      	\scalebox{0.8}{{\footnotesize
\begin{tikzpicture}[scale=.7]
\tikzstyle{inhibitorred}=[o-,draw=red,thick]
\tikzstyle{inhibitorblu}=[o-,draw=blue,thick]
\tikzstyle{pre}=[<-,thick]
\tikzstyle{post}=[->,thick]
\tikzstyle{prered}=[<-,thick,draw=red]
\tikzstyle{postred}=[->,thick,draw=red]
\tikzstyle{readblue}=[-, draw=blue,thick]
\tikzstyle{transition}=[rectangle, draw=black,thick,minimum size=5mm]
\tikzstyle{revtransition}=[rectangle, draw=red!80,fill=red!30,thick,minimum size=5mm]
\tikzstyle{place}=[circle, draw=black,thick,minimum size=5mm]
\tikzstyle{placeblu}=[circle, draw=blue!80,fill=blue!20,thick,minimum size=5mm]
\node[place,tokens=1] (p1) at (0,5)  {};
\node[place,tokens=1] (p2) at (1.5,5)  {};
\node[place,tokens=1] (p3) at (3,5)  {};
\node[place]          (p4) at (1.5,2.5) {};
\node[place]          (p5) at (.5,0) {};
\node[placeblu]       (p6) at (0,2.5) {};
\node[placeblu]       (p7) at (3,2.5) {};
\node[placeblu]       (p8) at (-1,0) {};
\node[place]          (p9) at (2.5,0) {};
\node[placeblu]       (p10) at (4,0) {};
\node[transition] (a) at (.5,3.75)  {$a$}
edge[pre] (p1)
edge[pre] (p2)
edge[post] (p6)
edge[post](p4);
\node[transition] (b) at (2.5,3.75)  {$b$}
edge[post] (p4)
edge[post] (p7)
edge[pre](p3)
edge[pre](p2);
\node[transition] (c) at (0.5,1.25) {$c$}
edge[pre] (p4)
edge[post] (p8)
edge[post] (p5);
\node[transition] (d) at (2.5,1.25) {$d$}
edge[pre] (p4)
edge[post] (p9)
edge[post] (p10);
\end{tikzpicture}
}}
	\caption{$\subnet{N}{T\setminus \Er}{}$}
         \label{fig:unetdentruno}
\end{subfigure}
\hspace{.7cm}
\begin{subfigure}[b]{0.30\textwidth}
	\centering
      	\scalebox{0.8}{{\footnotesize
\begin{tikzpicture}[scale=.7]
\tikzstyle{inhibitorred}=[o-,draw=red,thick]
\tikzstyle{inhibitorblu}=[o-,draw=blue,thick]
\tikzstyle{pre}=[<-,thick]
\tikzstyle{post}=[->,thick]
\tikzstyle{prered}=[<-,thick,draw=red]
\tikzstyle{postred}=[->,thick,draw=red]
\tikzstyle{readblue}=[-, draw=blue,thick]
\tikzstyle{transition}=[rectangle, draw=black,thick,minimum size=5mm]
\tikzstyle{revtransition}=[rectangle, draw=red!80,fill=red!30,thick,minimum size=5mm]
\tikzstyle{place}=[circle, draw=black,thick,minimum size=5mm]
\tikzstyle{placeblu}=[circle, draw=blue!80,fill=blue!20,thick,minimum size=5mm]
\node[place,tokens=1] (p1) at (0,5)  {};
\node[place,tokens=1] (p2) at (1.5,5)  {};
\node[place,tokens=1] (p3) at (3,5)  {};
\node[place]          (p4) at (1.5,2.5) {};
\node[place]          (p5) at (1.5,0) {};
\node[placeblu]       (p6) at (0,2.5) {};
\node[placeblu]       (p7) at (3,2.5) {};
\node[placeblu]       (p8) at (0,0) {};
\node[transition] (a) at (.5,3.75)  {$a$}
edge[pre] (p1)
edge[pre] (p2)
edge[post] (p6)
edge[post](p4);
\node[revtransition] (ar) at (-1,3.75)  {$\un{a}$}
edge[postred] (p1)
edge[postred] (p2)
edge[prered] (p6)
edge[prered](p4);
\node[transition] (b) at (2.5,3.75)  {$b$}
edge[post] (p4)
edge[post] (p7)
edge[pre](p3)
edge[pre](p2);
\node[revtransition] (br) at (4,3.75)  {$\un{b}$}
edge[prered] (p4)
edge[prered] (p7)
edge[postred](p3)
edge[postred](p2);
\node[transition] (c) at (1.5,1.25) {$c$}
edge[pre] (p4)
edge[post] (p8)
edge[post] (p5);
\node[revtransition] (c) at (0,1.25) {$c$}
edge[postred] (p4)
edge[prered] (p8)
edge[prered] (p5);
\end{tikzpicture}
}}
        \caption{$\revnet{N'}$}
        \label{fig:unetrev}
\end{subfigure}
\caption{An \rcn{} $N$ with reversing transitions $\Er = \setenum{\un{a},\un{c}}$, the \unet{}
$\subnet{N}{T\setminus \Er}{}$ and the complete
$\revnet{N'}$ associated to the net $N'$ in \Cref{fig:unetcomp}}\label{fig:unetvariantsbis}
\end{figure}

\begin{figure}[thb]
\centering
\begin{subfigure}[b]{0.35\textwidth}
	\centering
      	\scalebox{0.8}{{\footnotesize
\begin{tikzpicture}[scale=.7]
\tikzstyle{inhibitorred}=[o-,draw=red,thick]
\tikzstyle{inhibitorblu}=[o-,draw=blue,thick]
\tikzstyle{pre}=[<-,thick]
\tikzstyle{post}=[->,thick]
\tikzstyle{prered}=[<-,thick,draw=red]
\tikzstyle{postred}=[->,thick,draw=red]
\tikzstyle{readblue}=[-, draw=blue,thick]
\tikzstyle{transition}=[rectangle, draw=black,thick,minimum size=5mm]
\tikzstyle{revtransition}=[rectangle, draw=red!80,fill=red!30,thick,minimum size=5mm]
\tikzstyle{place}=[circle, draw=black,thick,minimum size=5mm]
\tikzstyle{placeblu}=[circle, draw=blue!80,fill=blue!20,thick,minimum size=5mm]
\node[place,tokens=1] (p1) at (0,5)  {};
\node[place,tokens=1] (p2) at (1.5,5)  {};
\node[place,tokens=1] (p3) at (3,5)  {};
\node[place]          (p4) at (1.5,2.5) {};
\node[place]          (p5) at (1.5,0) {};
\node[transition] (a) at (.5,3.75)  {$a$}
edge[pre] (p1)
edge[pre] (p2)
edge[post](p4);
\node[revtransition] (ar) at (-1,3.75)  {$\un{a}$}
edge[postred] (p1)
edge[postred] (p2)
edge[prered](p4);
\node[transition] (b) at (2.5,3.75)  {$b$}
edge[post] (p4)
edge[pre](p3)
edge[pre](p2);
\node[revtransition] (br) at (4,3.75)  {$\un{b}$}
edge[prered] (p4)
edge[postred](p3)
edge[postred](p2);
\node[transition] (c) at (1.5,1.25) {$c$}
edge[pre] (p4)
edge[post] (p5);
\node[revtransition] (c) at (0,1.25) {$c$}
edge[postred] (p4)
edge[prered] (p5);
\end{tikzpicture}
}}
	\caption{$N$}
         \label{fig:revunetnoncompl}
\end{subfigure}
\hspace{.7cm}
\begin{subfigure}[b]{0.35\textwidth}
	\centering
      	\scalebox{0.8}{{\footnotesize
\begin{tikzpicture}[scale=.7]
\tikzstyle{inhibitorred}=[o-,draw=red,thick]
\tikzstyle{inhibitorblu}=[o-,draw=blue,thick]
\tikzstyle{pre}=[<-,thick]
\tikzstyle{post}=[->,thick]
\tikzstyle{prered}=[<-,thick,draw=red]
\tikzstyle{postred}=[->,thick,draw=red]
\tikzstyle{readblue}=[-, draw=blue,thick]
\tikzstyle{transition}=[rectangle, draw=black,thick,minimum size=5mm]
\tikzstyle{revtransition}=[rectangle, draw=red!80,fill=red!30,thick,minimum size=5mm]
\tikzstyle{place}=[circle, draw=black,thick,minimum size=5mm]
\tikzstyle{placeblu}=[circle, draw=blue!80,fill=blue!20,thick,minimum size=5mm]
\node[place] (p1) at (0,5)  {};
\node[place] (p2) at (1.5,5)  {};
\node[place,tokens=1] (p3) at (3,5)  {};
\node[place,tokens=1]          (p4) at (1.5,2.5) {};
\node[place]          (p5) at (1.5,0) {};
\node[transition] (a) at (.5,3.75)  {$a$}
edge[pre] (p1)
edge[pre] (p2)
edge[post](p4);
\node[revtransition] (ar) at (-1,3.75)  {$\un{a}$}
edge[postred] (p1)
edge[postred] (p2)
edge[prered](p4);
\node[transition] (b) at (2.5,3.75)  {$b$}
edge[post] (p4)
edge[pre](p3)
edge[pre](p2);
\node[revtransition] (br) at (4,3.75)  {$\un{b}$}
edge[prered] (p4)
edge[postred](p3)
edge[postred](p2);
\node[transition] (c) at (1.5,1.25) {$c$}
edge[pre] (p4)
edge[post] (p5);
\node[revtransition] (c) at (0,1.25) {$c$}
edge[postred] (p4)
edge[prered] (p5);
\end{tikzpicture}
}}
	\caption{$N$ after executing $a$}
         \label{fig:unetrevnoncomplbis}
\end{subfigure}
\caption{An \rcn{} $N$ with reversing transitions $\Er = \setenum{\un{a},\un{c}}$ and the net $N$ after the firing of the transition $a$.}\label{fig:unetvariantstris}
\end{figure}

To clarify the crucial role played by  key-places, consider the \unet{} $N$ depicted \Cref{fig:unet}.
Simply adding the reversing transitions in accordance with Proposition~\ref{pr:constructing-a-reversible-unravel-net} would yield the net shown in \Cref{fig:revunetnoncompl}. However, this net is not an \rcn{}, as the net obtained by removing the reversing transitions is not complete. Now, consider the marking after firing $a$, as depicted in
\Cref{fig:unetrevnoncomplbis}. With this marking, the reversing transition $\un{b}$ is enabled and can be executed, contradicting the expectation that a transition can only be reversed  if it has been  previously executed.
The inclusion of the key-places for transitions $a, b$ and $c$ resolves  this problem.

\section{CCS processes as unravel nets}\label{sec:coding}
\subsection{Encoding of CCS processes}

\begin{figure*}[t]\label{fig:prefixing}
\begin{subfigure}[b]{0.25\textwidth}
	\centering
      	\scalebox{0.8}{\begin{tikzpicture}
\tikzstyle{inhibitorred}=[o-, draw=red,thick]
\tikzstyle{inhibitorblu}=[o-, draw=blue,thick]
\tikzstyle{pre}=[<-,thick]
\tikzstyle{post}=[->,thick]
\tikzstyle{readblue}=[-, draw=blue,thick]
\tikzstyle{transition}=[rectangle, draw=black,thick,minimum size=5mm]
\tikzstyle{place}=[circle, draw=black,thick,minimum size=5mm]
\node[place,tokens=1] (p1) at (0,5) [label=left:$\nil$] {};
\end{tikzpicture}}
         \caption{$\ccstonet{\nil}$}
         \label{fig:nil}
\end{subfigure}
\quad
\begin{subfigure}[b]{0.25\textwidth}
	\centering
      	\scalebox{0.8}{\begin{tikzpicture}
\tikzstyle{inhibitorred}=[o-, draw=red,thick]
\tikzstyle{unravel}=[->, draw=blue,thick]
\tikzstyle{pre}=[<-,thick]
\tikzstyle{post}=[->,thick]
\tikzstyle{readblue}=[-, draw=blue,thick]
\tikzstyle{transition}=[rectangle, draw=black,thick,minimum size=5mm]
\tikzstyle{place}=[circle, draw=black,thick,minimum size=5mm]
\tikzstyle{placeblu}=[circle, draw=blue!80,fill=blue!20,thick,minimum size=5mm]

\node[place,tokens=1] (p1) at (0,0) [label=above:$b.\nil$] {};
\node[place] (p2) at (2,0) [label=right:$\pastdec{ b}.\nil$] {};
\node[placeblu] (p3) at (2,1) [label=right:$\pastdec{ b}.\underline{b}$] {};

\node[transition] (b) at (1,0) {$b$}
edge[pre] (p1)
edge[post](p2)
edge[post](p3)

;

\end{tikzpicture}}
         \caption{$\ccstonet{b.\nil}$}
         \label{fig:b}
\end{subfigure}
\quad
\begin{subfigure}[b]{0.3\textwidth}
	\centering
      	\scalebox{0.8}{\begin{tikzpicture}
\tikzstyle{inhibitorred}=[o-, draw=red,thick]
\tikzstyle{unravel}=[->, draw=blue,thick]
\tikzstyle{pre}=[<-,thick]
\tikzstyle{post}=[->,thick]
\tikzstyle{readblue}=[-, draw=blue,thick]
\tikzstyle{transition}=[rectangle, draw=black,thick,minimum size=5mm]
\tikzstyle{place}=[circle, draw=black,thick,minimum size=5mm]
\tikzstyle{placeblu}=[circle, draw=blue!80,fill=blue!20,thick,minimum size=5mm]

\node[place,tokens=1] (p1) at (0,0) [label=above:$a.b.\nil$] {};
\node[place,] (p2) at (2,0) [label=above:$ \pastdec{a}.b.\nil$] {};
\node[placeblu] (p3) at (2,1) [label=right:$\hat{a}.\underline{a}$] {};
\node[place,] (p4) at (4,0) [label=right:$\pastdec{a}.\pastdec{ b}.\nil$] {};
\node[placeblu] (p5) at (4,1) [label=right:$\pastdec{ a}.\pastdec{ b}.\underline{b}$] {};

\node[transition] (a) at (1,0) {$a$}
edge[pre] (p1)
edge[post](p2)
edge[post](p3)
;

\node[transition] (b) at (3,0) {$\hat{a}.b$}
edge[pre] (p2)
edge[post] (p4)
edge[post] (p5)
;
\end{tikzpicture}}
         \caption{$\ccstonet{a.b.\nil}$}
         \label{fig:ab}
\end{subfigure}
     \caption{Example of nets corresponding to CCS processes}
\end{figure*}
We now recall the encoding of CCS terms into Petri nets due to Boudol and
Castellani~\cite{BCIC94}. It is worth noting that the original encoding was on
\textit{proved terms} instead of plain CCS. The difference between proved
terms and CCS is that  in a proved term the labels carry the position
of the process which did the action. Hence, we will use \emph{decorated}
versions of labels. For instance, $\pastdec a.b$ denotes an event $b$ that has
been preceded by the occurrence of $a$  in the term $a.b$. Analogously,
labels carry also information about the syntactical structure of a term,
actions corresponding to subterms of a choice and of a parallel composition
are also decorated with an index $i$ that indicates the subterm that performs
the action. An interesting aspect of this encoding is that these information
is reflected in the name of the places and the transitions of the nets, which
simplifies the formulation of the behavioural correspondence of a term and its
associated net. We write $\lbl{\_}$ for the function that removes decorations
for a name, e.g., $\lbl{\pastdec{a}.\pastdec{b}.c} = c$.

We are now in place to define and comment the encoding of a CCS term into a
net. The encoding is inductively defined on the structure of the CCS process.
For a CCS process $P$, its encoded net is
$\netBCsimple P = \langle S_P, T_P, F_P, {\sf m}_P\rangle$. The net
corresponding to the inactive process $\nil$, is just a net with just one
marked place and with no transition, that is:
\begin{defi}\label{bc:zeronet}
  The net
  $\zeronetBC = \langle \setenum{\nil}, \emptyset, \emptyset,
  \setenum{\nil}\rangle$ is the net associated to $\nil$ and it is called
  \emph{zero}.
\end{defi}
\noindent 
To ease notation in the constructions we are going to present, we adopt
the following conventions: let $X \subseteq S\cup T$ be a set of places and
transitions, we write $\pastdec{\alpha}.X$ for the set
$\{\pastdec{\alpha}. x\ |\ x \in X\}$ containing the {\em decorated} versions of
places and transitions in $X$. Analogously we lift this notation to relations:
if $R$ is a binary relation on $(S\cup T)$, then
$\pastdec \alpha.R = \{(\pastdec \alpha.x, \pastdec \alpha.y)\ |\ (x,y)\in
R\}$ is a binary relation on $(\alpha.S\cup \alpha.T)$.

The net $\netBCsimple{\alpha.P}$ corresponding to a process $\alpha.P$ extends
$\netBCsimple{P}$ with two extra places $\alpha.P$ and
$\newinenc{\hat{a}.\underline{\alpha}}$ and one transition $\alpha$. The place
$\alpha.P$ stands for the process that executes the prefix $\alpha$ and
follows as $P$. The place $\newinenc{\hat{a}.\underline{\alpha}}$ is not in
the original encoding of~\cite{BCIC94}; we have add it to ensure that the
obtained net is complete, which is essential for the definition of the
reversible net.
This will become clearer when commenting the encoding of the parallel composition.
It should be noted that this addition does not interfere with the behaviour of
the net, since all added places are final. Also a new transition, named
$\alpha$ is created and added to the net, and the flow relation is updated
accordingly. 
We use colours to indicate the name of places that serve
  as key places. The input and output places of the added transitions vary
  depending on the CCS operator under consideration.

Figures \ref{fig:nil}, \ref{fig:b} and \ref{fig:ab} report
respectively the encodings of the inactive process, of the process $b.\nil$
and $a.b.\nil$. Moreover the aforementioned figures systematically show how
the prefixing operator is rendered into Petri nets. As a matter of fact, the
net $a.b.\nil$ is built starting from the net corresponding to $b.\nil$ by
adding the prefix $a$.
We note that also the label of transitions is affected by appending the label
of the new prefix at the beginning. This is rendered in \Cref{fig:ab} where
the transition mimicking the action $b$ is labeled as $\hat{a}.b$ indicating
that an $a$ was done before $b$. In what follows we will often omit such
representation from figures.

\begin{defi}\label{bc:prefixnet}
  Let $P$ a CCS process and $\netBCsimple{P} = \macronet{P}{}$ be the
  associated net. Then $\prefnetBC \alpha P$ is the net
  $\macronet{{\alpha.P}}{}$ where
  \[
    \begin{array}{rcl}
      S_{\alpha.P} & = & \{\alpha.P, \newinenc{\pastdec{\alpha}.\underline{\alpha}}\} \cup  \pastdec \alpha.S_P \\
      T_{\alpha.P} & = & \{\alpha\} \cup \pastdec \alpha. T_p \\
      F_{\alpha.P} & = &\{(\alpha.P, \alpha),\newinenc{(\alpha,\pastdec{\alpha}.\underline{\alpha}) } \} \cup
                         \{ (\alpha,\pastdec \alpha.b) \ |\ b\in \marko{m}_{P} \}  \cup \pastdec \alpha.F_P\\
      \marko{m}_{\alpha.P} & = & \setenum{\alpha.P} \\
    \end{array}
  \]
\end{defi}

\noindent 
The set of \emph{key}-places of $\prefnetBC \alpha P$ is
$\pastdec
\alpha.\keypl{T_{P}}\cup\setenum{\newinenc{\pastdec{\alpha}.\underline{\alpha}}}$,
where $\keypl{T_{P}}$ are the \emph{key}-places of $\netBCsimple{P}$.

As we have done for prefixes, for a set $X$ of transitions and places we write $\pardec i X$ for
$\{\pardec i x \ |\ x\in X\}$, which straightforwardly lifts to relations.
We do the same with  $\plusdec i$ and $\restdec a$, which are the decorations for the sum and
the restriction.

The encoding of parallel goes along the line of the prefixing one. Also in
this case we have to decorate the places (and transitions) with the position
of the term in the syntax tree. To this end, each branch of the parallel is
decorated with $\pardec i$ with {\sf i} being the $i$-th position. Regarding the
transitions, we have to add all the possible synchronisations among the
processes in parallel. This is why, along with the transitions of the branches
(properly decorated with $\pardec i$) we have to add extra transitions to
indicate the possible synchronisation. Naturally a synchronisation is possible
when one label is the co-label of the other transition. Figure \ref{fig:par}
shows the net corresponding to the process $a.b \parallel \co{a}.c$. As we can
see, the encoding builds upon the encoding of $a.b$ and $\co{a}.c$, by (i)
adding to all the places and transitions whether the branch is the left one or
the right one and (ii) adding an extra transition and place for the only
possible synchronisation.
We add an extra place (in line with the prefixes) to mark the fact that a
synchronisation has taken place. Let us note that the extra places
$\underline{a}$, $\underline{\overline{a}}$ and $\underline{\tau}$ are used to
understand whether the two prefixes have done a partial synchronisation
 or they contributed to do a synchronisation. Suppose, for
example, that the net had not such places, and suppose that we have two tokens
in the places $\pardec 0\hat{a}.b$ and $\pardec 1\hat{\bar{ a}}.b$. Now,
how can we understand whether these two tokens are the result of the firing
sequence $a$,$\overline{a}$ or they are the result of the $\tau$ transition?
It is impossible, but by using the aforementioned extra-places, which are
instrumental to tell if a single prefix has executed, we can distinguish the
$\tau$ from the sequence $a\overline{a}$ and then reverse
accordingly.

\begin{defi}\label{bc:parnet}
  Let $\netBCsimple {P_1}$ and $\netBCsimple {P_2}$ be the nets associated to
  the processes $P_1$ and $P_2$. Then $\parnetBC {P_1} {P_2}$ is the net
  $\macronet{{P_1 \| P_2}}{}$ where
  \[
    \begin{array}{rcl}
      S_{P_1 \| P_2}
      & =
      & \pardec 0 {S_{P_1}}  \cup  \pardec 1 {S_{P_2}} \cup
        \setcomp{s_{\setenum{t,t'}}}{t\in T_{P_1} \wedge t'\in T_{P_2} \wedge
        \co{\lbl t} = \lbl {t'}}
      \\ [5pt]
      T_{P_1 \| P_2}
      & =
      & \pardec 0 {T_{P_1}}  \cup  \pardec 1 {T_{P_2}} \cup
        \{\setenum{t,t'}\ |\ t\in T_{P_1} \wedge t'\in T_{P_2} \wedge
        \co{\lbl t} = \lbl {t'}\}
      \\  [5pt]
      F_{P_1 \| P_2}
      & =
      & \pardec 0 F_{P_1} \cup \pardec 1 F_{P_2}        \\  [3pt]
&&
       \cup \setcomp{(\setenum{t,t'},s_{\setenum{t,t'}})}{t\in T_{P_1} \wedge \, t'\in T_{P_2}
        \, \wedge\,
        \co{\lbl t} = \lbl {t'}  }
      \\  [3pt]
      &
      &
        \cup \setcomp{(\pardec i s,\setenum{t_1,t_2})}{(s,t_i)\in F_{P_i}} \cup
        \setcomp{(\setenum{t_1,t_2},\pardec i s)}{(t_i,s)\in F_{P_i}\
        \land\ s\not\in\keypl{T_{P_{i}}}}
      \\  [5pt]
      \marko{m}_{P_1 \| P_2}
      & =
      &\pardec 0 \marko{m}_{P_1} \cup \pardec 1 \marko{m}_{P_2}
    \end{array}
  \]
\end{defi}
\noindent 
The \emph{key}-places of the resulting net are the following.
\[
  \pardec 0 \keypl{T_{P_1}} \cup \pardec 1 \keypl{T_{P_2}} \cup
  \{s_{\setenum{t,t'}}\ |\ t\in T_{P_1} \wedge t'\in T_{P_2} \wedge \co{\lbl
    t} = \lbl {t'}\}
\]
They are obtained by properly renaming the ones arising from the encoding of
the branches and those corresponding to the synchronisations of the
components.

The encoding of the choice operator is similar to the parallel one. The only
difference is that we do not have to deal with possible synchronisations since
the branches of a choice are mutually exclusive. \Cref{fig:plus}
illustrates the net corresponding to the process $a.b + \bar{a}.c$. As in the
previous examples, the net is built upon the subnets representing $a.b$ and
$\bar{a}.c$.

\begin{defi}\label{bc:choicenet}
  Let $\netBCsimple {P_i}$ be the net associated to the processes $P_i$ for
  $i\in I$. Then $\choicenetBCm{i\in I}{P_i}$ is the net
  $\macronet{+_{i\in I}P_i}{}$ where:
  \[
    \begin{array}{rcl}
      S_{+_{i\in I}P_i} & = & \cup_{i\in I}\plusdec i {S_{P_i}} \\
      T_{+_{i\in I}P_i} & = & \cup_{i\in I}\plusdec i {T_{P_i}} \\
      F_{+_{i\in I}P_i} & = & \setcomp{(\plusdec i x,\plusdec i y)}{(x,y)\in F_{P_i}}\cup
                              \setcomp{(\plusdec j s,\plusdec i t) }{s \in \marko{m}_{P_j}
                              \land \pre{t}\in \marko{m}_{P_i} \land i\neq j}\\
      \marko{m}_{+_{i\in I}P_i} & = & \cup_{i\in I}\plusdec i {\marko{m}_{P_i}}. \\
    \end{array}
  \]
\end{defi}
\noindent 
In this case the \emph{key}-places of $\choicenetBCm{i\in I}{P_i}$ are just
the union of all \emph{key}-places after the suitable renaming, i.e.,
$\cup_{i\in I}\plusdec i {\keypl{T_{P_i}}}$.

\begin{figure}[t]
  \begin{subfigure}[b]{0.45\textwidth}
    \centering
    \scalebox{0.8}{\begin{tikzpicture}[scale=.8]
\tikzstyle{inhibitorred}=[o-, draw=red,thick]
\tikzstyle{inhibitorblu}=[o-, draw=blue,thick]
\tikzstyle{pre}=[<-,thick]
\tikzstyle{post}=[->,thick]
\tikzstyle{readblue}=[-, draw=blue,thick]
\tikzstyle{transition}=[rectangle, draw=black,thick,minimum size=5mm]
\tikzstyle{place}=[circle, draw=black,thick,minimum size=5mm]
\tikzstyle{placeblu}=[circle, draw=blue!80,fill=blue!20,thick,minimum size=5mm]

\node[place,tokens=1] (p1) at (0,5) [label=left:$\parallel_0a.b$] {};
\node[place,tokens=1] (p2) at (3,5) [label=right:$\parallel_1\overline{a}.c$] {};
\node[place] (p3) at (0,2.5) [label=left:$\parallel_0\pastdec{a}.b$] {};
\node[placeblu] (p7) at (-1.5,3.75) [label=left:$\underline{a}$] {};
\node[placeblu] (p8) at (4.5,3.75) [label=right:$\underline{\overline{a}}$] {};
\node[placeblu](p11) at (1.5,2.5) [label=below:$\underline{\tau}$] {};

\node[placeblu] (p9) at (-1.5,1.25) [label=left:$\underline{b}$] {};
\node[placeblu] (p10) at (4.5,1.25) [label=right:$\underline{c}$] {};

\node[place] (p4) at (3,2.5) [label=right:$\parallel_1\hat{\bar{ a}}.c$] {};
\node[place] (p5) at (0,0) [label=left:$\parallel_0\pastdec{a}.\pastdec b$] {};
\node[place] (p6) at (3,0) [label=right:$\parallel_1\hat{\bar a}.\pastdec c$] {};
\node[transition] (a) at (0,3.75)  {$a$}
edge[pre] (p1)
edge[post](p3)
edge[post](p7)
;

\node[transition] (b) at (0,1.25) {$b$}
edge[pre] (p3)
edge[post] (p5)
edge[post] (p9)

;

\node[transition] (bara) at (1.5,3.75)  {$\tau$}
edge[pre] (p1)
edge[pre] (p2)
edge[post](p3)
edge[post](p4)
edge[post](p11);

\node[transition] (bara) at (3,3.75)  {$\bar{a}$}
edge[pre] (p2)
edge[post](p4)
edge[post] (p8)
;

\node[transition] (d) at (3,1.25)  {$c$}
edge[pre] (p4)
edge[post](p6)
edge[post](p10);
\end{tikzpicture}}
    \caption{$\ccstonet{a.b \para \bar{a}.c}$}
    \label{fig:par}
  \end{subfigure}
  \begin{subfigure}[b]{0.45\textwidth}
    \centering
    \scalebox{0.8}{\begin{tikzpicture}[scale=.8]
\tikzstyle{inhibitorred}=[o-, draw=red,thick]
\tikzstyle{inhibitorblu}=[o-, draw=blue,thick]
\tikzstyle{pre}=[<-,thick]
\tikzstyle{post}=[->,thick]
\tikzstyle{readblue}=[-, draw=blue,thick]
\tikzstyle{transition}=[rectangle, draw=black,thick,minimum size=5mm]
\tikzstyle{place}=[circle, draw=black,thick,minimum size=5mm]
\tikzstyle{placeblu}=[circle, draw=blue!80,fill=blue!20,thick,minimum size=5mm]

\node[place,tokens=1] (p1) at (0,5) [label=left:$+_0a.b$] {};
\node[place,tokens=1] (p2) at (3,5) [label=right:$+_1\co{a}.c$] {};
\node[place] (p3) at (0,2.5) [label=left:$+_0\pastdec{a}.b$] {};
\node[place] (p4) at (3,2.5) [label=right:$+_1\hat{\bar{ a}}.c$] {};
\node[place] (p5) at (0,0) [label=left:$+_0\pastdec{a}.\pastdec{b}$] {};
\node[place] (p6) at (3,0) [label=right:$+_1\hat{\bar{ a}}.\pastdec c$] {};

\node[placeblu] (p9) at (-1.5,1.25) [label=left:$\underline{b}$] {};
\node[placeblu] (p10) at (4.5,1.25) [label=right:$\underline{c}$] {};

\node[placeblu] (p7) at (-1.5,3.75) [label=left:$\underline{a}$] {};
\node[placeblu] (p8) at (4.5,3.75) [label=right:$\underline{\overline{a}}$] {};

\node[transition] (a) at (0,3.75)  {$a$}
edge[pre] (p1)
edge[pre] (p2)
edge[post] (p7)
edge[post](p3);

\node[transition] (b) at (0,1.25) {$b$}
edge[pre] (p3)

edge[pre] (p3)
edge[post] (p9)

edge[post] (p5)
;

\node[transition] (bara) at (3,3.75)  {$\bar{a}$}
edge[pre] (p1)
edge[pre] (p2)
edge[post](p8)
edge[post](p4);

\node[transition] (d) at (3,1.25)  {$c$}
edge[pre] (p4)
edge[post](p10)
edge[post](p6);
\end{tikzpicture}}
    \caption{$\ccstonet{a.b + \bar{a}.c}$}
    \label{fig:plus}
  \end{subfigure}
  \caption{Example of nets corresponding to CCS parallel and choice operator. We omit the trailing $\nil$}
\end{figure}

We write $T^{a}$ for the set all transitions in $T$ labelled by $a$ or $\bar{a}$, i.e.,
$\{t\in T\ |\ {\lbl t} = a\ \lor\ {\lbl t} = \bar{a}\}$.
The encoding of the hiding operator simply removes all transitions whose
labels corresponds to actions performed over the restricted name and the key-places associated to these
transitions.

\begin{defi}\label{bc:restnet}
  Let $P$ a CCS process and $\netBCsimple{P} = \macronet{P}{}$ be the
  associated net. Then $\restnetBC P a$ is the net $\macronet{P\setminus a}{}$
  where
  \[
    \begin{array}{rcl}
      S_{P\setminus a}
      & =
      & \restdec a (S_P \setminus \keypl{T^{a}})
      \\
      T_{P\setminus a}
      & =
      & \restdec a (T_P \setminus T^{a})
      \\
      F_{P\setminus a}
      & =
      & \setcomp{(\restdec a s,\restdec a t)}{s\in S_P \setminus \keypl{T^{a}}\ \land\ (s,t)\in F_{P}
      \ \land\  t\not\in T^{a}}\ \cup\ 
      \\
      &
      & \setcomp{(\restdec a t,\restdec a s)}{s\in S_P \setminus \keypl{T^{a}}\ \land\ (t,s)\in F_{P}
      \ \land\ t\not\in T^{a}}
      \\
      \marko{m}_{P\setminus a}
      & =
      &   \restdec a \marko{m}_{P}
      \\
    \end{array}
  \]
\end{defi}
\noindent 
In this case, as the number of \emph{firable} transitions decreases, a
corresponding decrease is observed in the number of \emph{key}-places. Hence,
$\keypl{\restdec a (T_P \setminus T^{a})} = \restdec a (\keypl{T_P} \setminus
\keypl{T^{a}})$.
\Cref{fig:restr} shows the net corresponding to the CCS process
  $(a.b \parallel \co{a}.c)\setminus a$. Observe that certain transitions are
  removed (those labeled with the restricted action), along with their
  associated key-places, although this is not strictly necessary. In fact,
  after the removal of the transitions, the respective places remain isolated because the only connected transitions have been removed.

\begin{figure}[t]
  \centering
  \scalebox{0.9}{\begin{tikzpicture}[scale=.8]
\tikzstyle{inhibitorred}=[o-, draw=red,thick]
\tikzstyle{inhibitorblu}=[o-, draw=blue,thick]
\tikzstyle{pre}=[<-,thick]
\tikzstyle{post}=[->,thick]
\tikzstyle{readblue}=[-, draw=blue,thick]
\tikzstyle{transition}=[rectangle, draw=black,thick,minimum size=5mm]
\tikzstyle{place}=[circle, draw=black,thick,minimum size=5mm]
\tikzstyle{placeblu}=[circle, draw=blue!80,fill=blue!20,thick,minimum size=5mm]

\node[place,tokens=1] (p1) at (0,5) [label=left:$\setminus_a\parallel_0a.b$] {};
\node[place,tokens=1] (p2) at (3,5) [label=right:$\setminus_a\parallel_1\overline{a}.c$] {};
\node[place] (p3) at (0,2.5) [label=left:$\setminus_a\parallel_0\pastdec{a}.b$] {};
\node[placeblu](p11) at (1.5,2.5) [label=below:$\setminus_a \underline{\tau}$] {};

\node[placeblu] (p9) at (-1.5,1.25) [label=left:$\setminus_a\underline{b}$] {};
\node[placeblu] (p10) at (4.5,1.25) [label=right:$\setminus_a\underline{c}$] {};

\node[place] (p4) at (3,2.5) [label=right:$\setminus_a\parallel_1\hat{\bar{ a}}.c$] {};
\node[place] (p5) at (0,0) [label=left:$\setminus_a\parallel_0\pastdec{a}.\pastdec b$] {};
\node[place] (p6) at (3,0) [label=right:$\setminus_a\parallel_1\hat{\bar a}.\pastdec c$] {};

\node[transition] (b) at (0,1.25) {$b$}
edge[pre] (p3)
edge[post] (p5)
edge[post] (p9)

;

\node[transition] (bara) at (1.5,3.75)  {$\tau$}
edge[pre] (p1)
edge[pre] (p2)
edge[post](p3)
edge[post](p4)
edge[post](p11);


\node[transition] (d) at (3,1.25)  {$c$}
edge[pre] (p4)
edge[post](p6)
edge[post](p10);
\end{tikzpicture}}
  \caption{The net $\restnetBC{(a.b \parallel \co{a}.c)}{a}$}
  \label{fig:restr}
\end{figure}

\begin{figure}[b]
  \centering
  \scalebox{0.9}{\begin{tikzpicture}[scale=.9]
\tikzstyle{inhibitorred}=[o-, draw=red,thick]
\tikzstyle{inhibitorblu}=[o-, draw=blue,thick]
\tikzstyle{pre}=[<-,thick]
\tikzstyle{post}=[->,thick]
\tikzstyle{epost}=[->, draw=magenta,thick]
\tikzstyle{placeblu}=[circle, draw=blue!80,fill=blue!20,thick,minimum size=5mm]

\tikzstyle{readblue}=[-, draw=blue,thick]
\tikzstyle{transition}=[rectangle, draw=black,thick,minimum size=5mm]
\tikzstyle{place}=[circle, draw=black,thick,minimum size=5mm]
\tikzstyle{eplace}=[circle, draw=magenta,thick,minimum size=5mm]

\node[place,tokens=1] (p1) at (-1,5) [label=left:$\parallel_0a.a$] {};
\node[place,tokens=1] (p2) at (2,3.75) [label=above:$\parallel_1 +_0 \overline{a}$] {};
\node[place] (p3) at (-1,2.5) [label=left:$\parallel_0\pastdec a.a$] {};
\node[placeblu] (p8) at (-2.5,3.75) [label=left:$\underline{a}$] {};

\node[place] (p4) at (2,1.25) [label=below:$\parallel_1 +_0 \hat{\bar{ a}}$] {};
\node[placeblu] (p10) at (3,1.25) [label=below:$\underline{a}$] {};
\node[placeblu] (p11) at (5,2.5) [label=right:$\underline{b}$] {};
\node[placeblu] (p12) at (.5,5) [label=right:$\underline{\tau}$] {};
\node[placeblu] (p13) at (.5,0) [label=right:$\underline{\tau}$] {};

\node[place] (p5) at (-1,0) [label=below:$\parallel_0\pastdec a.\pastdec a$] {};
\node[placeblu] (p9) at (-2.5,1.25) [label=left:$\underline{a}$] {};

\node[place, tokens=1] (p6) at (4,3.75) [label=right:$\parallel_1 +_1 b$] {};
\node[place] (p7) at (4,1.25) [label=below:$\parallel_1 +_1 \pastdec b$] {};

\node[transition] (a1) at (-1,3.75)  {$a$}
edge[pre] (p1)
edge[post](p3)
edge[post](p8)
;

\node[transition] (a2) at (-1,1.25) {$a$}
edge[pre] (p3)
edge[post] (p5)
edge[post](p9)

;

\node[transition] (a3) at (.5,1.25) {$\tau$}
edge[pre](p3)
edge[pre] (p2)
edge[post] (p5)
edge[post] (p4)
edge[post] (p13)
;

\node[transition] (tau) at (.5,3.75)  {$\tau$}
edge[pre] (p1)
edge[pre] (p2)
edge[post] (p3)
edge[post] (p4)
edge[post] (p12);

\node[transition] (bara) at (2,2.5)  {$\bar{a}$}
edge[pre] (p2)
edge[pre] (p2)
edge[pre](p6)
edge[post](p4)
edge[post](p10)

;

\node[transition] (b) at (4,2.5)  {$b$}
edge[pre] (p6)
edge[pre] (p2)
edge[post](p7)
edge[post](p11)
;

\end{tikzpicture}}
  \caption{A complex example: $\ccstonet{a.a \parallel \co{a} + b}$}
  \label{fig:complex}
\end{figure}

In Figure \ref{fig:complex}, a more complex example is depicted, illustrating
the net corresponding to the process $a.a \parallel \co{a} + b$. In this case,
the process on the right of the parallel composition can synchronise with the
one on the left one in two different occasions.
This is why there are two different transitions representing the
synchronisations.
However, due to the nature of the process on the right-hand side being a
choice, there is a possibility that the right branch of that choice gets
executed, thereby preventing the synchronisation from occurring. As the right
branch of the parallel constitutes a choice between two options, the encoding
designates these branches as `$\pardec{1}\plusdec{0}$' and `$\pardec{1}\plusdec{0}$'
respectively. These labels serve to identify the left and right branches of
the choice, which is situated within the right branch of the parallel
operator.

The following proposition is instrumental for the main correspondence result.

\begin{prop}\label{pr:allareun}
  The nets defined in
  \Cref{bc:zeronet,bc:prefixnet,bc:parnet,bc:choicenet,bc:restnet} are {\em
    complete} unravel nets.
\end{prop}
\begin{proof}
  By induction on the structure of a CCS process. Clearly the net $\zeronetBC$
  is an unravel net and it is trivially complete because it has no transition.
  Assume now that $\netBCsimple{P} = \macronet{P}{}$ associated with the CCS
  process $P$ is a complete \unet. Also
  $\prefnetBC \alpha P = \macronet{{\alpha.P}}{}$ is an \unet{} as it is
  obtained by adding a new transition $\alpha$ that precedes all transitions
  in $T_P$. Moreover, a new key-place
  $\newinenc{\pastdec{\alpha}.\underline{\alpha}}$ is added for such
  transition. Assuming now that $\netBCsimple{P_1}$ and $\netBCsimple{P_2}$
  are the two complete \unet{s} associated with $P_1$ and $P_2$. The net
  $\parnetBC {P_1} {P_2}$ is an \unet{} as the two components, when
  \emph{synchronise}, have the effect of the local changes beside the
  key-places.
  For each synchronising transition $\setenum{t, t'}$, a corresponding
  key-place $s_{\setenum{t, t'}}$ exists, rendering the net complete.
  Similarly, $\choicenetBCm{i\in I}{P_i}$ is a complete unravel net, as each
  $\netBCsimple{P_i}$ is a complete unravel net. The additional flow arcs
  ensure that only transitions of a specific component are executed. Lastly,
  $\restnetBC P a$ is complete because the elimination of transitions does not
  add any new behaviour.
\end{proof}

\subsection{Encoding of RCCS processes}

We are now at the point where we can define the network that corresponds to an
RCCS process. So far, our focus has been on encoding CCS processes into nets.
Since RCCS is built upon CCS processes, our encoding of RCCS naturally builds
upon the encoding of CCS. To do so, we first introduce the concept of
ancestor, i.e., the initial process from which an RCCS process is derived.
Notably, in the context of our discussion involving coherent RCCS processes
(as defined in \cref{lbl:initial_proc}), an RCCS process invariably possesses
an ancestor.

The ancestor $\ancestor{R}$ of an RCCS process $R$ can be calculated through
syntactical analysis of $R$, as all information about its past is stored
within memories. The sole instance in which a process must wait for its
counterpart is during a memory fork, denoted as $\entry{1}$ or $\entry{2}$.

\begin{defi}\label{lbl:ancestor}
  Given a coherent RCCS process $R$, its ancestor $\ancestor{R}$ is derived by
  using the inference rules of \Cref{fig:ancestor}. The rules use the
  pre-congruence relation $\precon$ defined as $\con$ (see Figure~\ref{fig:struct})
  with the exception that rule \textsc{Split} can be only applied from right
  to left.
\end{defi}

\begin{figure}[t]
  \begin{mathpar}
    \inferrule*[Right={\scriptsize(\textsc{Act})}]
    {}
    {\mem{\_}
      {\alpha^{\newinenc z}_z}
      {\sum_{i\in I \setminus \{z\}}\alpha_i.P_i}{m}
      \proc P \anc m \proc\sum_{i\in I}\alpha_i.P_i}
    \\
    \inferrule*[Right={\scriptsize(\textsc{Pre})}]
    {R \precon R' \and R' \anc S' \and S' \precon S}{R \anc S}
    \and
    \inferrule*[Right={\scriptsize(\textsc{Par})}]
    {R\anc R' }
    {R \parallel S \anc R' \parallel S}
    \\
    \inferrule*[Right={\scriptsize(\textsc{Res})}]
    {R\anc R' }{R\hide{a} \anc R' \hide{a}}
    \and
    \inferrule*[Right={\scriptsize(\textsc{Init})}]
    {R\anc^* \emp \proc P}{\ancestor{R} = P}
  \end{mathpar}
  \caption{Ancestor inference rules}
  \label{fig:ancestor}
\end{figure}

\begin{exa}
Consider the RCCS term $R$ below:
  \begin{align*}
    R = \ \
    & \lpar{\event{m_2}{a^1}{\nil} \cdot \lpar{\emp}}\proc b \ \parallel\
      \rpar{\event{m_2}{a^1}{\nil} \cdot \lpar{\emp}}\proc c \ \parallel\
    \\
    & \event{m_1}{\co a^1}{\nil} \cdot \lpar{\rpar{\emp}}\proc \nil
      \ \parallel\ \rpar{\rpar{\emp}}\proc d
  \end{align*}
  with $m_1 = \lpar{\emp}$ and $m_2 = \lpar{\rpar{\emp}}$
  By applying the inference rules in \Cref{fig:ancestor}, we compute its ancestor as follows:
  \begin{align*}
  	R \anc\, & \event{m_2}{a^1}{\nil}\cdot \lpar{\emp}\proc (b \parallel c ) \parallel  & (\textsc{Act}) \\
	& \event{m_1}{\co a^1}{\nil} \cdot \lpar{\rpar{\emp}}\proc \nil
      \ \parallel\ \rpar{\rpar{\emp}}\proc d  \\
      \anc\, & \lpar{\emp}\proc a.(b\parallel c) \parallel \event{m_1}{\co a^1}{\nil} \cdot \lpar{\rpar{\emp}}\proc \nil \parallel\ \rpar{\rpar{\emp}}\proc d & (\textsc{Act})\\
            \anc\, & \lpar{\emp}\proc a.(b\parallel c) \parallel \lpar{\rpar{\emp}}\proc  \co{a} \parallel\ \rpar{\rpar{\emp}}\proc d & (\textsc{Pre})\\
                        \anc\, & \lpar{\emp}\proc a.(b\parallel c) \parallel \rpar{\emp}\proc ( \co{a} \parallel d)  & (\textsc{Pre})\\
             \anc \,& \emp \proc ( a.(b\parallel c) \parallel \co{a} \parallel d)
  \end{align*}
  that is $\ancestor{R} = a.(b\parallel c) \parallel \co{a} \parallel d$. Let
  us note that reversing a synchronisation is achieved by applying the
  \textsc{Act} rule twice---each monitored process can undo its respective
  part of the synchronisation. This is possible due to the coherence of
  processes. Essentially, upon encountering a synchronisation event, a process
  possesses adequate information to revert to its previous local state.
  Conversely, when encountering a split event, the process must await its
  siblings (as per the \textsc{Pre} rule) to reconstruct the parallel process.
  \end{exa}

\begin{lem}\label{lem:ancestor}
  For any coherent RCCS process $R$ its ancestor $\ancestor{R}$ exists and it is
  unique.
\end{lem}
\begin{proof}
  Since $R$ is a coherent process then there exists a CCS process $P$ such
  that $\emp \proc P \rccst{\,}^* R$. By Property \ref{prp:fw_only}
  we have that $\emp \proc P \rightarrow^*R$, and by
  applying Corollary \ref{cor:loop} we obtain that $ R \bk^{*} \emp \proc P$.
  The proof is then by induction on the number $n$ of reductions contained in $\bk^{*}$
  and by noticing that for each application of $\bk$ there exists a corresponding
  rule of $\anc$.
\end{proof}

\noindent 
There is a tight correspondence between RCCS memories and transitions/places
names. That is, a memory contains all the information to recover the path from
the root to the process itself. To this end, we introduce the function
$\cammino{\cdot}$, which is inductively defined as follows
\begin{align*}
  & \cammino{m\cdot\event{m'}{\alpha^{ i}}{\nil}}
    = \cammino{m\cdot\event{*}{\alpha^{ i}}{\nil}}
    = \pastdec{\alpha}.\cammino{m}
  \\
  & \cammino{m\cdot\event{m'}{\alpha^{ i}}{Q}}
    = \cammino{m\cdot\event{*}{\alpha^{ i}}{Q}}
    = \plusdec{i}\pastdec{\alpha}.\cammino{m}
  \\
  &\cammino{m\cdot\entry{i}}
    = \pardec {i}\,\cammino{m}
  \\
  & \cammino{\emp} = \epsilon
\end{align*}

\begin{exa}\label{ex:rccs_processess}
  Let us consider the RCCS processes $R_1$ and $R_2$ defined below
    \begin{align*}
      R_1 = &
       \mem{*}{a^1}{\nil}{\lpar\emp} \proc b \para \rpar\emp\proc
        \bar{a}.c
      \\ R_2  =
      & \mem{*}{b^1}{\nil}{\mem{m_2}{a^1}{\nil}{\lpar\emp}} \proc
        \nil \para \mem{m_1}{\co{a}^1}{\nil}{\rpar\emp}\proc c
    \end{align*}
    with $m_i = \entry{i}\cdot \emp$. Their corresponding nets are shown in
    Figure~\ref{fig:ex_para}.

\begin{figure}[t]
  \begin{subfigure}[b]{0.5\textwidth}
    \centering
    \scalebox{0.9}{{\scriptsize
\begin{tikzpicture}[scale=.8]
\tikzstyle{inhibitorred}=[o-,draw=red,thick]
\tikzstyle{inhibitorblu}=[o-,draw=blue,thick]
\tikzstyle{pre}=[<-,thick]
\tikzstyle{post}=[->,thick]
\tikzstyle{prered}=[<-,thick,draw=red]
\tikzstyle{postred}=[->,thick,draw=red]
\tikzstyle{readblue}=[-, draw=blue,thick]
\tikzstyle{transition}=[rectangle, draw=black,thick,minimum size=5mm]
\tikzstyle{revtransition}=[rectangle, draw=red!80,fill=red!30,thick,minimum size=5mm]
\tikzstyle{place}=[circle, draw=black,thick,minimum size=5mm]
\tikzstyle{placeblu}=[circle, draw=blue!80,fill=blue!20,thick,minimum size=5mm]
\node[place] (p1) at (-0.75,5) [label=above:$\parallel_0a.b$] {};
\node[place,tokens=1] (p2) at (3,5) [label=above:$\parallel_1\overline{a}.c$] {};
\node[place,tokens=1]          (p3) at (-0.75,2.5) [label={[xshift=.6cm, 
                                           yshift=-0.8cm]$\parallel_0\pastdec{a}.b$}] {};
\node[place]          (p4) at (3,2.5) [label={[xshift=-.6cm, 
                                        yshift=-0.8cm]$\parallel_1\hat{\bar{ a}}.c$}] {};
\node[place]          (p5) at (-0.75,0) [label=right:$\parallel_0\pastdec{a}.\pastdec b$] {};
\node[place]          (p6) at (3,0) [label=left:$\parallel_1\hat{\bar a}.\pastdec c$] {};
\node[placeblu,tokens=1]       (p7) at (-2,2.5) [label=below:$\underline{a}$] {};
\node[placeblu]       (p8) at (4.25,2.5) [label=below:$\underline{\overline{a}}$] {};
\node[placeblu]       (p9) at (-2,0) [label=right:$\underline{\overline{b}}$] {};
\node[placeblu]       (p10) at (4.25,0) [label=left:$\underline{c}$] {};
\node[placeblu]       (p11) at (1.15,2) [label=below:$\underline{\tau}$] {};
\node[transition] (a) at (-0.75,3.75)  {$a$}
edge[pre] (p1)
edge[post](p3)
edge[post](p7);
\node[revtransition] (ar) at (-2,3.75)  {$\underline{a}$}
edge[postred, bend left] (p1)
edge[prered](p3)
edge[prered](p7);
\node[transition] (b) at (-0.75,1.25) {$b$}
edge[pre] (p3)
edge[post] (p5)
edge[post] (p9);
\node[revtransition] (br) at (-2,1.25) {$\underline{b}$}
edge[postred] (p3)
edge[prered] (p5)
edge[prered] (p9);
\node[transition] (tau) at (.5,3.75)  {$\tau$}
edge[pre] (p1)
edge[pre, bend left] (p2)
edge[post](p3)
edge[post](p11)
edge[post, bend right](p4);
\node[revtransition] (taur) at (1.8,3.75)  {$\tau$}
edge[postred, bend right] (p1)
edge[postred] (p2)
edge[prered, bend left](p3)
edge[prered](p11)
edge[prered](p4);
\node[transition] (aa) at (3,3.75)  {$\bar{a}$}
edge[pre] (p2)
edge[post](p4)
edge[post] (p8);
\node[revtransition] (aar) at (4.25,3.75)  {$\underline{\bar{a}}$}
edge[postred, bend right] (p2)
edge[prered](p4)
edge[prered] (p8);
\node[transition] (c) at (3,1.25)  {$c$}
edge[pre] (p4)
edge[post](p6)
edge[post](p10);
\node[revtransition] (cr) at (4.25,1.25)  {$\underline{c}$}
edge[postred] (p4)
edge[prered](p6)
edge[prered](p10);
\end{tikzpicture}
}}
    \caption{$\revnet{\ccstonet{R_1}}$}
    \label{fig:reva}
  \end{subfigure}
  \begin{subfigure}[b]{0.5\textwidth}
    \centering
    \scalebox{0.9}{{\scriptsize
\begin{tikzpicture}[scale=.8]
\tikzstyle{inhibitorred}=[o-,draw=red,thick]
\tikzstyle{inhibitorblu}=[o-,draw=blue,thick]
\tikzstyle{pre}=[<-,thick]
\tikzstyle{post}=[->,thick]
\tikzstyle{prered}=[<-,thick,draw=red]
\tikzstyle{postred}=[->,thick,draw=red]
\tikzstyle{readblue}=[-, draw=blue,thick]
\tikzstyle{transition}=[rectangle, draw=black,thick,minimum size=5mm]
\tikzstyle{revtransition}=[rectangle, draw=red!80,fill=red!30,thick,minimum size=5mm]
\tikzstyle{place}=[circle, draw=black,thick,minimum size=5mm]
\tikzstyle{placeblu}=[circle, draw=blue!80,fill=blue!20,thick,minimum size=5mm]
\node[place] (p1) at (-0.75,5) [label=above:$\parallel_0a.b$] {};
\node[place] (p2) at (3,5) [label=above:$\parallel_1\overline{a}.c$] {};
\node[place]          (p3) at (-0.75,2.5) [label={[xshift=.6cm, 
                                           yshift=-0.8cm]$\parallel_0\pastdec{a}.b$}] {};
\node[place,tokens=1]          (p4) at (3,2.5) [label={[xshift=-.6cm, 
                                        yshift=-0.8cm]$\parallel_1\hat{\co{ a}}.c$}] {};
\node[place,tokens=1]          (p5) at (-0.75,0) [label=right:$\parallel_0\pastdec{a}.\pastdec b$] {};
\node[place]          (p6) at (3,0) [label=left:$\parallel_1\hat{\bar a}.\pastdec c$] {};
\node[placeblu]       (p7) at (-2,2.5) [label=below:$\underline{a}$] {};
\node[placeblu]       (p8) at (4.25,2.5) [label=below:$\underline{\overline{a}}$] {};
\node[placeblu,tokens=1]       (p9) at (-2,0) [label=right:$\underline{\overline{b}}$] {};
\node[placeblu]       (p10) at (4.25,0) [label=left:$\underline{c}$] {};
\node[placeblu,tokens=1]       (p11) at (1.15,2) [label=below:$\underline{\tau}$] {};
\node[transition] (a) at (-0.75,3.75)  {$a$}
edge[pre] (p1)
edge[post](p3)
edge[post](p7);
\node[revtransition] (ar) at (-2,3.75)  {$\underline{a}$}
edge[postred, bend left] (p1)
edge[prered](p3)
edge[prered](p7);
\node[transition] (b) at (-0.75,1.25) {$b$}
edge[pre] (p3)
edge[post] (p5)
edge[post] (p9);
\node[revtransition] (br) at (-2,1.25) {$\underline{b}$}
edge[postred] (p3)
edge[prered] (p5)
edge[prered] (p9);
\node[transition] (tau) at (.5,3.75)  {$\tau$}
edge[pre] (p1)
edge[pre, bend left] (p2)
edge[post](p3)
edge[post](p11)
edge[post, bend right](p4);
\node[revtransition] (taur) at (1.8,3.75)  {$\tau$}
edge[postred, bend right] (p1)
edge[postred] (p2)
edge[prered, bend left](p3)
edge[prered](p11)
edge[prered](p4);
\node[transition] (aa) at (3,3.75)  {$\bar{a}$}
edge[pre] (p2)
edge[post](p4)
edge[post] (p8);
\node[revtransition] (aar) at (4.25,3.75)  {$\underline{\bar{a}}$}
edge[postred, bend right] (p2)
edge[prered](p4)
edge[prered] (p8);
\node[transition] (c) at (3,1.25)  {$c$}
edge[pre] (p4)
edge[post](p6)
edge[post](p10);
\node[revtransition] (cr) at (4.25,1.25)  {$\underline{c}$}
edge[postred] (p4)
edge[prered](p6)
edge[prered](p10);
\end{tikzpicture}
}}
    \caption{$\revnet{\ccstonet{R_2}}$}
    \label{fig:revtaub}
  \end{subfigure}
  \caption{Example of nets corresponding to RCCS process $R_1$ and $R_2$}
  \label{fig:ex_para}
\end{figure}

  We have that the path of the left process is
  $\cammino{ \mem{*}{a^1}{\nil}{\lpar\emp}} = \pardec{0}\pastdec{a}$, while
  the path of the right process is $\cammino{\rpar\emp} = \pardec{1}$.
\end{exa}
The encoding of an RCCS process should yield an equivalent net to that of its
ancestor, with the only potential distinction being the marking – indicating
the specific locations where tokens are placed.
And such positions are inferred from the information stored in memories.
Following the intuitions in \cref{sec:coding}, we will treat names of places
and transitions as strings. When we write $\phi X$, where $X$ is a set of
strings and $\phi \in \{\pardec{i},\plusdec{i},\pastdec{\alpha},\restdec{a}\}$,
we are indicating the set $\{\phi x \st x\in X\}$.
Also, we will indicate with $\tilde{\phi}$ the sequence $\phi_1 \cdots \phi_n$
with $\phi_i \in \{\pardec{j},\plusdec{j },\pastdec{\alpha},\restdec{a}\}$.
Then the \emph{marking} function $\mrkof{\cdot}$ is inductively defined as follows:
\begin{align*}
  \mrkof{R \parallel S} =\
  & \mrkof{R} \fuse \mrkof{S}
  \\
  \mrkof{R \backslash a} =\
  & \restdec a\mrkof{R}
  \\
  \mrkof{m\cdot\event{m_1}{\alpha^{ i}}{\nil}\cdot \emp \proc P} =\
  &
    \{\alpha, m_1\} \cup \pastdec{\alpha}.\mrkof{m\cdot \emp \proc P}
  \\
  \mrkof{m\cdot\event{m_1}{\alpha^{ i}}{Q}\cdot \emp \proc P} =\
  &
    \{ \plusdec{i}\alpha, m_1\} \cup \plusdec{i}\pastdec{\alpha}.\mrkof{m\cdot \emp \proc P}
  \\
  \mrkof{m\cdot\event{*}{\alpha^{ i}}{\nil}\cdot \emp \proc P} =\
  & \{ \pastdec{\alpha}.\underline{\alpha}\} \cup \pastdec{\alpha}.\mrkof{m\cdot \emp \proc P}
  \\
  \mrkof{m\cdot\event{*}{\alpha^{ i}}{Q}\cdot \emp \proc P} =\
  &
    \{ \plusdec{i}\pastdec{\alpha}.\underline{\alpha}\} \cup
    \plusdec{i}\pastdec{\alpha}.\mrkof{m\cdot \emp \proc P}
  \\
  \mrkof{ m\cdot \entry{i}\cdot \emp \proc P} =\
  &
    \pardec{i}\,\mrkof{m\cdot \emp \proc P}
  \\
  \mrkof{\emp \proc P} =\
  &
    \{P\}
\end{align*}
where $\fuse$ is defined as the usual set union on single element, and as the
merge on pairs of the form $\{t_1,m_2\}$ $\{t_2,m_1\}$ where
$\{t_1,m_2\}\fuse\{t_2,m_1\} = s_{\{t_1,t_2\}}$ if
${\lbl {t_1}} = \co{\lbl {t_2}}$ and $t_i = \cammino{m_i}\alpha_i$ with $\alpha_i = \lbl {t_i}$,
where $s_{\{t_1,t_2\}}$ is the key place of the synchronisation between  transitions
$t_1$ and $t_2$.

\begin{exa}
  Let us consider the RCCS processes $R_1$ and $R_2$ of Example~\ref{ex:rccs_processess}.
  The marking of the process $R_1$ is
  \begin{align*}
    & \mrkof{\mem{*}{a^1}{\nil}{\lpar\emp} \proc b \para \rpar\emp\proc \bar{a}.c}
    \\
    =\
    &
      (\mrkof{\mem{*}{a^1}{\nil}{\lpar\emp} \proc b}) \fuse
      (\pardec{1}\mrkof{ \emp \proc \bar{a}.c} )
    \\
    =\
    &
      (\pardec{0}\mrkof{\mem{*}{a^1}{\nil}  {\emp}\proc b} ) \fuse (\{\pardec{1}\bar{a}.c\} )
    \\
    =\
    &
      (\{\pardec{0}\pastdec{a}.\underline{a}\}
      \cup \pardec{0}\pastdec{a}.\mrkof{\emp\proc b}) \fuse \{\pardec{1}\bar{a}.c\}
    \\
    =\
    &
      \{\pardec{0}\pastdec{a}.\underline{a}, \pardec{0}\pastdec{a}.b\}
      \fuse \{\pardec{1}\bar{a}.c\}
    \\
    =\
    &
      \{\pardec{0}\pastdec{a}.\underline{a}, \pardec{0}\pastdec{a}.b,\pardec{1}\bar{a}.c\}
  \end{align*}
  and the marking of the  process $R_2$ is
  \begin{align*}
    &
      \mrkof{\mem{*}{b^1}{\nil}{\mem{m_2}{a^1}{\nil}{\lpar\emp}} \proc \nil
    \para \mem{m_1}{\co{a}^1}{\nil}{\rpar\emp}\proc c}
    \\
    = \
    &
      (\mrkof{\mem{*}{b^1}{\nil}{\mem{m_2}{a^1}{\nil}{\lpar\emp}} \proc \nil} )
      \fuse (\mrkof{\mem{m_1}{\co{a}^1}{\nil}{\rpar\emp}\proc c})
    \\
    = \
    &
      (\pardec{0}\mrkof{\mem{*}{b^1}{\nil}{\mem{m_2}{a^1}{\nil}{\emp}} \proc \nil} )
      \fuse
      (\pardec{1}\mrkof{\mem{m_1}{\co{a}^1}{\nil}{\emp}\proc c})
    \\
    =\
    &
      \pardec{0} (\{a,m_2\},\pastdec{a}.\mrkof{\mem{*}{b^1}{\nil}{\emp}\proc \nil})
      \fuse (\pardec{1}\{ \{\co{a},m_1\}, \hat{\co{a}}.\mrkof{{\emp}\proc c}\})
    \\
    =\
    & \pardec{0} (\{a,m_2\},\pastdec{a}.\{\underline{b},b\} )
      \fuse (\pardec{1}\{ \{\co{a},m_1\}, \hat{\co a}.c\})
    \\
    =\
    &
      \{\{\pardec{0}a,m_2\}, \pardec{0}\pastdec{a}.\underline{b},
      \pardec{0}\pastdec{a}.\pastdec{b}\}
      \fuse
      \{\{\pardec{1}\co{a},m_1\}, \pardec{1}\hat{\co{a}}.c\}
    \\
    =\
    &
      \{\{\pardec{0}a, \pardec{1}\co{a}\},\pardec{0}\pastdec{a}.\underline{b},
      \pardec{0}\pastdec{a}.\pastdec{b}, \pardec{1}\hat{\co a}.c \}
  \end{align*}
\end{exa}
\noindent 
We are now in place to define a property that relates the definitions of
$\mrkof{\cdot}$ and $\cammino{\cdot}$ with RCCS processes.

\begin{pty}\label{prp:memory_dec}
Let $R = m\proc \sum_{i\in I}\alpha_i.P_i$ be a RCCS process. For any $z\in I$
such that $R\ltsk{m}{\alpha_z} \mem{*}{\alpha^{z}_z}{\sum_{i\in I\setminus\{z\}}\alpha_i.P_i}{m}\proc P_z$ we have that

\begin{align*}
  \mrkof{\mem{*}{\alpha^{z}_z}{\sum_{i\in I\setminus\{z\}}\alpha_i.P_i}{m}\proc P_z} =\
  &
    \mrkof{R} \setminus \{ \cammino{m}\plusdec{z}\alpha_z.P_z\}
  \\
  &
    \cup\  \{\cammino{m}\plusdec{z}\pastdec{\alpha_z}.\underline{\alpha_z}, \cammino{m}\plusdec{z}\alpha_z.P_z \}
\end{align*}

\end{pty}
\begin{proof}
The proof is by induction on the size of $m$. The base case with $m = \emp$  trivially holds.
In the inductive case we have $m = m_1 \cdot e\cdot \emp$ where $e$ can be
$\entry{i}$, $\event{*}{\beta^{ i}}{Q}$ or $\event{m_2}{\beta^{ i}}{Q}$.
We will show the first two cases, with the third being similar to the second one.
We have that
$$ S_0 = m_1\cdot \emp\proc \sum_{i\in I}\alpha_i.P_i\ltsk{m_1\cdot \emp}{\alpha_z} \mem{*}{\alpha^{z}_z}{\sum_{i\in I\setminus\{z\}}\alpha_i.P_i}{m_1\cdot \emp}\proc P_z = R_0$$
and by applying inductive hypothesis (on a shorter memory) we have that
\begin{align}\label{lbl:induction}
\mrkof{S_0} = \mrkof{R_0} \setminus \{ \cammino{m_1}\plusdec{z}\alpha_z.P_z\} \cup \{\cammino{m_1}\plusdec{z}\pastdec{\alpha_z}.\underline{\alpha_z}, \cammino{m_1}\plusdec{z}\alpha_z.P_z \}
\end{align}
We proceed by case analysis.
\begin{description}
\item[$e = \entry{i}$] let us note that $\cammino{m} = \pardec{i}\cammino{m_1}$, and that
$\mrkof{R} = \pardec{i}\mrkof{R_0}$ and $\mrkof{S} = \pardec{i}\mrkof{S_0}$. Thanks to \cref{lbl:induction} we know the form of $\mrkof{S_0}$, hence
\begin{align*}
  \mrkof{S}  =\
  & \pardec{i}\mrkof{S_0}
  \\
  =\
  & \pardec{i}\mrkof{R_0} \setminus \{ \pardec{i}\cammino{m_1}\plusdec{z}\alpha_z.P_z\} \cup \{\pardec{i}\cammino{m_1}\plusdec{z}\pastdec{\alpha_z}.\underline{\alpha_z}, \pardec{i}\cammino{m_1}\plusdec{z}\alpha_z.P_z \}  \\
  =\
  &\mrkof{R} \setminus \{ \cammino{m}\plusdec{z}\alpha_z.P_z\} \cup \{\cammino{m}\plusdec{z}\pastdec{\alpha_z}.\underline{\alpha_z}, \cammino{m}\plusdec{z}\alpha_z.P_z \}
\end{align*}
as desired.
\item[$e = \event{*}{\beta^{ i}}{Q}$] let us note that
$\cammino{m} = \plusdec{i}\pastdec{\beta}.\cammino{m_1}$, and that
$\mrkof{R} = \plusdec{i}\pastdec{\beta}.\mrkof{R_0} \cup \{\plusdec{i}\pastdec{\beta}.\underline{\beta}\}$,
and $\mrkof{S} = \plusdec{i}\pastdec{\beta}.\mrkof{S_0} \cup \{\plusdec{i}\pastdec{\beta}.\underline{\beta}\}$.
  Thanks to \cref{lbl:induction} we know the form of $\mrkof{S_0}$, hence
\begin{align*}
  \mrkof{S} = \
  &
    \plusdec{i}\pastdec{\beta}.\mrkof{S_0}
    \cup \{\plusdec{i}\pastdec{\beta}.\underline{\beta}\}
  \\
  = \
  & \plusdec{i}\pastdec{\beta}.\{\mrkof{R_0} \setminus \{ \cammino{m_1}\plusdec{z}\alpha_z.P_z\}
  \\
  &
    \cup\ \{\cammino{m_1}\plusdec{z}\pastdec{\alpha_z}.\underline{\alpha_z},
    \cammino{m_1}\plusdec{z}\alpha_z.P_z \} \}
    \ \cup\  \{\plusdec{i}\pastdec{\beta}.\underline{\beta}\}
  \\
  = \
  & \plusdec{i}\pastdec{\beta}.\mrkof{R_0}
    \setminus \{ \plusdec{i}\pastdec{\beta}.\cammino{m_1}\plusdec{z}\alpha_z.P_z\}
  \\
  & \cup \{
    \plusdec{i}\pastdec{\beta}.\cammino{m_1}\plusdec{z}
    \pastdec{\alpha_z}.\underline{\alpha_z},
    \plusdec{i}\pastdec{\beta}.\cammino{m_1}\plusdec{z}\alpha_z.P_z \} \}
   \cup \{\plusdec{i}\pastdec{\beta}.\underline{\beta}\}
  \\
  =\
  &
    \plusdec{i}\pastdec{\beta}.\mrkof{R_0} \setminus \{ \cammino{m}\plusdec{z}\alpha_z.P_z\}
  \\
  &
    \cup
    \{\cammino{m}\plusdec{z}\pastdec{\alpha_z}.\underline{\alpha_z},
    \cammino{m}\plusdec{z}\alpha_z.P_z \} \}
    \cup\{\plusdec{i}\pastdec{\beta}.\underline{\beta}\}
  \\
  = \
  &\mrkof{R} \setminus \{ \cammino{m}\plusdec{z}\alpha_z.P_z\} \cup \{\cammino{m}\plusdec{z}\pastdec{\alpha_z}.\underline{\alpha_z},\cammino{m}\plusdec{z}\alpha_z.P_z \} \}
\end{align*}
as desired.\qedhere
\end{description}
\end{proof}
\noindent 
As a consequence, we have the following corollary.
\begin{cor}\label{cor:mem_dec}
Let $R = m\proc \alpha.P$ be a RCCS process. For any $z\in I$
such that $R\ltsk{m}{\alpha_z} \mem{*}{\alpha}{\nil}{m}\proc P$ we have that
$$\mrkof{\mem{*}{\alpha}{\nil}{m}\proc P} =
\mrkof{R} \setminus \{ \cammino{m}\alpha.P\} \cup \{\cammino{m}\pastdec{\alpha}.\underline{\alpha}, \cammino{m}\alpha.P \}  $$
\end{cor}
\noindent 
We are now ready to formalise the reversible net corresponding to an RCCS process.

\begin{defi}\label{de:rccsasun}
  Let $R$ be an RCCS term with $\ancestor{R} = P$. Then
  $\revnet{\ccstonet{R}}$ is the net $\langle S,T, F, \mrkof R\rangle$ where
  ${\ccstonet{P}} = \macronet{}{}$.
\end{defi}
\noindent 
Note that the reversible net corresponding to a coherent RCCS process $R$
retains identical places, transitions, and flow relationships as the ancestor
of R. The sole divergence lies in the marking, which is derived through the
utilisation of the computational history stored within the memories of $R$.
The following Proposition is a consequence of the Proposition~\ref{pr:allareun},
Lemma~\ref{lem:ancestor} and of the definition of $\mrkof{\cdot}$.

\begin{prop}\label{pr:rccsasun}
  Let $R$ be an RCCS term with $\ancestor{R} = P$. Then
  $\revnet{\ccstonet{R}}$ is a reversible unravel net.
\end{prop}

\subsection{Correctness result}

We prove the correctness of our encoding in terms of a behavioural
equivalence. To this aim we reformulate the definition of \emph{forward and
  reverse bisimilarity}~\cite{ccsk}, initially stated for CCSK, to cope with
RCCS terms and Petri nets.

\begin{defi}[Forward and reverse bisimulation]\label{def:bisim}
  Let $R$ a coherent RCCS process and $N = \langle S,T, F, \marko m\rangle$ an
  \rcn{}. The relation $\mathcal{R}$ is a forward reverse bisimulation if
  whenever $(R,N) \in \mathcal{R}$:
    \begin{enumerate}
    \item if $R \ltsk{m}{\alpha} R'$ then there exist $t \in T$ and
      $\marko{m '}$ such that $\marko{m} \trans{t} \marko{m'}$,
      $t = (\cammino{m}\alpha,\fe)$ and
      $(R',\langle S,T, F, \marko m'\rangle)\in \mathcal{R}$;
    \item if $R \rltsk{m}{\alpha} R'$ then there exist $t \in T$ and
      $\marko{m '}$ such that $\marko m \trans{t} \marko{m'}$,
      $t = (\cammino{m}\alpha,\re)$ and
      $(R',\langle S,T, F, \marko m'\rangle)\in \mathcal{R}$;
    \item if $R \ltsk{m_1,m_2}{\tau} R'$ then there exist
      $(t_1,\fe), (t_2,\fe) \in T$ and $\marko{m '}$ such that
      $\marko m \trans{(\{t_1,t_2\},\fe)} \marko{m'}$,
      $\co{\lbl {t_1}} = \lbl {t_2}$, $\cammino{m_i} < t_i$ for $i\in \{1,2\}$
      and $(R',\langle S,T, F, \marko m'\rangle)\in \mathcal{R}$;
    \item if $R \rltsk{m_1,m_2}{\tau} R'$ then there exist
      $(t_1,\re), (t_2,\re) \in T$ and $\marko{m '}$ such that
      $\marko m \trans{(\{t_1,t_2\},\re)} \marko{m'}$,
      $\co{\lbl {t_1}} = \lbl {t_2}$, $\cammino{m_i} < t_i$ for $i\in \{1,2\}$
      and $(R',\langle S,T, F, \marko m'\rangle)\in \mathcal{R}$;
    \item if $\marko{m} \trans{t} \marko{m'}$ with
      $t = (\cammino{m}\alpha,\fe)$ then there exists $R, R'$ such that
      $\mrkof{R} = \marko{m}$, $\mrkof{R'} = \marko{m'}$,
      $R \ltsk{m}{\alpha} R'$ and
      $(R',\langle S,T, F, \marko m'\rangle)\in \mathcal{R}$;
    \item if $\marko{m} \trans{t} \marko{m'}$ with
      $t = (\cammino{m}\alpha,\re)$ then there exists $R, R'$ such that
      $\mrkof{R} = \marko{m}$, $\mrkof{R'} = \marko{m'}$,
      $R \rltsk{m}{\alpha} R'$ and
      $(R',\langle S,T, F, \marko m'\rangle)\in \mathcal{R}$;
    \item if $\marko m \trans{(\{t_1,t_2\},\fe)} \marko{m'}$ with
      $\co{\lbl {t_1}} = \lbl {t_2}$ and $\cammino{m_i}\alpha_i = t_i$ with $\lbl {t_i} = \alpha_i$ for $i\in \{1,2\}$
       then there exists $R, R'$ such that
      $\mrkof{R} = \marko{m}$, $\mrkof{R'} = \marko{m'}$,
      $R \ltsk{m_1,m_2}{\tau} R'$ and
      $(R',\langle S,T, F, \marko m'\rangle)\in \mathcal{R}$;
    \item if $\marko m \trans{(\{t_1,t_2\},\re)} \marko{m'}$ with
      $\co{\lbl {t_1}} = \lbl {t_2}$ and $\cammino{m_i}\alpha_i = t_i$ with $\lbl {t_i} = \alpha_i$ for $i\in \{1,2\}$ 
      then there exists $R, R'$ such that
      $\mrkof{R} = \marko{m}$, $\mrkof{R'} = \marko{m'}$,
      $R \rltsk{m_1,m_2}{\tau} R'$ and
      $(R',\langle S,T, F, \marko m'\rangle)\in \mathcal{R}$.
    \end{enumerate}
  The largest forward reverse bisimulation is called forward reverse
  bisimilarity, denoted with $\sim_{FR}$.
\end{defi}

We first prove that two  coherent RCCS processes which are structurally congruent are encoded within the
\emph{same} \rcn{}. Subsequently, we demonstrate the equivalence between a
step taken in the process algebra and the firing of an appropriate transition
in the corresponding network, and vice versa.

\begin{lem}[Preservation]\label{lem:struct}
Let $R_1$ and $R_2$ be two coherent RCCS processes. If $R_1 \con R_2$ then
$\revnet{\ccstonet{R_1}}$ and  $\revnet{\ccstonet{R_2}}$ are isomorphic and have the same marking
up to places renaming.
\end{lem}
\begin{proof}
  Since $\equiv$ is defined on monitored processes, then the only axiom which
  changes the structure of the ancestor process is the $\alpha$-renaming.
  Hence $R_1$ and $R_2$ have the same ancestor, say $P$, up to
  $\alpha$-renaming. It is easy to see that the two generated nets have the
  same places, transitions and flow relation up to renaming, hence they are
  isomorphic.
  We just have to check whether the initial markings are the same. The proof
  follows by induction and  case analysis on the last applied axiom of
  $\equiv$:

  \setlist[description]{font=\normalfont\scshape}
  \begin{description}
  \item[Split]   
If the last applied rule is (\textsc{Split}), w.l.o.g.~we can
    assume $R_1 = m\proc (P_1\parallel P_2)$ and
    $R_2 = \lpar{m}\proc P_1 \parallel \rpar{m}\proc P_2 $. We need to show
    that $\mrkof{R_1} = \mrkof{R_2}$. By looking at the definition of
    $\mrkof{\cdot}$ we have that
    $$\mrkof{R_1} = \marko{m} \text{ and }$$ and
    $$\mrkof{R_2} = \mrkof{ \lpar{m}\proc P_1} \fuse \mrkof{ \rpar{m}\proc P_2}$$
   Also, by definition of $\mrkof{\cdot}$ we have that $\marko{m} = \marko{m}_0 \cup \cammino{m}\marko{m}_{P_1 \parallel P_2}$, where $\marko{m}_{P_1 \parallel P_2}$ is the initial marking of the net encoding $(P_1 \parallel P_2)$.
   Hence, we can divide this marking into the marking of $P_1$ and the marking of $P_2$ as follows:
   $$\marko{m} = \marko{m}_0 \cup \cammino{m}\pardec{0}\marko{m}_{P_1} \cup \cammino{m}\pardec{1}\marko{m}_{P_2}$$
   Also, we have that:
   \begin{align*}
   &\mrkof{ \lpar{m}\proc P_1} = \marko{m}_a \cup \cammino{ \lpar{m}}\marko{m}_{P_1} =  \marko{m}^a_0 \cup \cammino{m}\pardec{0}\marko{m}_{P_1}\\
  & \mrkof{ \rpar{m}\proc P_2} = \marko{m}_b \cup \cammino{ \rpar{m}}\marko{m}_{P_2} = 
  \marko{m}^b_0 \cup \cammino{ m}\pardec{1}\marko{m}_{P_2}
   \end{align*}
Were $\marko{m}_a = \marko{m}_b = \marko{m}_0$, since it is the marking
derived from the information contained into the memory $m$. Now, it is simple
to conclude, since $\marko{m}_a \fuse \marko{m}_b = \marko{m}_0 \fuse \marko{m}_0 = \marko{m}_0$, as they are the same
marking and since we are considering reachable processes it is impossible for a process to synchronise in the future with itself, hence $\fuse$ acts as the normal set union.
  \item[Res] this case is a simplified version of the previous one.
  \item[$\alpha$] Suppose $\mrkof{R_1} = \marko{m} \cup \marko{m'}$ where
    $\marko{m'}$ is the markings containing the bound action which will be
    converted by the last application of $\equiv$. By inductive hypothesis we
    also have that $\mrkof{R_2} = \marko{m} \cup \alpha(\marko{m'})$ where the
    $\alpha$-conversion is applied only to those names which contains the
    bound action, that is $\marko{m'}$. We have that the two
    nets have the same marking up to some renaming, as desired.\qedhere
\end{description}
\end{proof}

\begin{lem}[Soundness]\label{lem:sound}
  Let $R_1$ be an RCCS coherent process and
  $\revnet{\ccstonet{R_1}}\!=\!\langle S, T, F, \mrkof{R_1}\rangle$ its
  corresponding \rcn{}. If $R_1 \rccst{\hat{m}:\alpha} R_2$ then
  \begin{itemize}

  \item  $\revnet{\ccstonet{R_2}}={\langle S, T, F, \mrkof{R_2}\rangle}$; and
  \item there exists $t \in T$ such that $\mrkof{R_1}\trans{t}\mrkof{R_2}$; and

  \item for some $d\in \setenum{\fe,\re}$ either
    \begin{itemize}
    \item
      $\hat{m} = m$ and $t = (\cammino{m}\alpha,d)$; or
    \item $\hat{m} =m_1,m_2$ and $\alpha = \tau$ and there exist two
      transitions $(t_1,d),(t_2,d) \in T$ with $\co{\lbl {t_1}} = \lbl {t_2}$ and
      $\cammino{m_i}\alpha_i = t_i$ with $\lbl {t_i} = \alpha_i$ for $i\in \{1,2\}$, and $t = (\setenum{t_1,t_2},d)$. %
    \end{itemize}
  \end{itemize}
\end{lem}

\begin{proof}
  As $R_1$ is a coherent process then it has an ancestor $\ancestor{R_1}$, say
  $P$ which is unique (thanks to~\cref{lem:ancestor}), which is the same ancestor of $R_2$, as $R_2$ is reached by $R_1$ with
  one reduction step. Therefore $\revnet{\ccstonet{R_1}}$ and
  $\revnet{\ccstonet{R_2}}$ have the same places, transitions and flow
  relation, the only difference being the marking. We show that for each move in
  the process algebra a corresponding  firing of a transition $t\in T$ exists such
  that $(\mrkof{R_1} \setminus \pre{t}) \cup \post{t} = \mrkof{R_2}$.

  We have
  two cases: either the process synchronises with the context or it performs a
  $\tau$~(or a reversing of any of them). Both cases are similar, so we will
  focus on the first one.
  We proceed by induction on the derivation $R_1\ltsk{m}{\alpha} R_2$ with a
  case analysis on the last applied rule. The base cases correspond to the
  application of either $\textsc{r-act}$ or $\textsc{r-act}\rev$.
  \setlist[description]{font=\normalfont\scshape}
    \begin{description}
\item[r-act]
  Consider the application of the rule $\textsc{r-act}$. We have
  \[
    R_1 = m\proc \sum_{i\in I}\alpha_i.Q_i\ltsk{m}{\alpha^{}_z}
    \mem{*}{\alpha^{z}_z}{\sum_{i\in I\setminus\{z\}}\alpha_i.Q_i}{m}\proc P_z
    = R_2
  \]
  We first consider the case where $|I| = 1$. Hence we have
  \[R_1 = m\proc \alpha.Q\ltsk{m}{\alpha^{}} \mem{*}{\alpha}{\nil}{m}\proc Q =
  R_2\]
  The marking corresponding to $R_1$ in the net $\revnet{\ccstonet{P}} = {\langle S, T, F, \marko{m}\rangle}$ is
  $\mrkof{R_1} = \mrkof{ m\proc \alpha.Q}$
  and thanks to \cref{cor:mem_dec} the marking of $R_2$ is
  \begin{align*}
    \mrkof{R_2} & = \mrkof{\mem{*}{\alpha}{\nil}{m}\proc Q} \\
    & = \mrkof{m\proc
      \alpha.Q} \setminus \{\cammino{m}\alpha.Q\}\cup
    \{\cammino{m}\pastdec{\alpha}.\underline{\alpha}\}\cup
    \{\cammino{m}\pastdec{\alpha}.\marko{m}_Q\}
  \end{align*}
  By construction (see \cref{bc:prefixnet}), the net $\revnet{\ccstonet{P}}$
  contains a transition $t \in T$ such that $t = (\cammino{m}\alpha,\fe)$, with
  $\pre{t} = \{\cammino{m}\alpha.Q\}$ and
  $\post{t} = \{\cammino{m}\pastdec{\alpha}.\underline{\alpha}\}\cup
  \{\cammino{m}\pastdec{\alpha}.\marko{m}_Q\}$.
  The thesis follows by observing that such transition is enabled at $\mrkof{R_1}$ because
  $\{\cammino{m}\alpha.Q\}\in \mrkof{R_1}$ by
  definition of $\mrkof{\cdot}$, and $\mrkof{R_1}\trans{t}\mrkof{R_2}$.

  Consider now the case with $|I| > 1$.
  \[
    R_1 = m\proc \sum_{i\in I}\alpha_i.Q_i\ltsk{m}{\alpha^{}_z}
    \mem{*}{\alpha^{z}_z}{\sum_{i\in I\setminus\{z\}}\alpha_i.Q_i}{m}\proc P_z
    = R_2
  \]
  The marking corresponding to $R_1$ in the net $\revnet{\ccstonet{P}} = {\langle S, T, F, \marko{m}\rangle}$ is
  \begin{align*}
    \mrkof{R_1} & = \mrkof{ m\proc \sum_{i\in I}\alpha_i.Q_i} \\
                & = \bigcup_{i\in I} \mrkof{ m\proc \alpha_i.Q_i}
  \end{align*} 
  and it contains the marked places $\setcomp{\cammino{m}\plusdec{i}\alpha_i.Q_i}{i\in I}$.
  Again by construction, the net $\revnet{\ccstonet{P}}$
  contains a transition $t \in T$ such that $t = (\cammino{m}\plusdec{z}\alpha_z,\fe)$, with $z\in I$,
  $\pre{t} = \setcomp{\cammino{m}\plusdec{i}\alpha_i.Q_i}{i\in I}$ and
  $\post{t} = \{\cammino{m}\plusdec{z}\pastdec{\alpha_z}.\underline{\alpha_z}\}\cup
  \{\cammino{m}\plusdec{z}\pastdec{\alpha_z}.\marko{m}_{Q_z}\}$ and again
  $\mrkof{R_1}\trans{t}\mrkof{R_2}$ where
   $\mrkof{R_2}$ is the marking 
   \[\qquad\quad\;\mrkof{m\proc\!\sum_{i\in I}\!\alpha_i.Q_i}\!\setminus\!\setcomp{\cammino{m}\plusdec{i}\alpha_i.Q_i\!}{\!i\!\in\!I}
   \cup\{\cammino{m}\plusdec{z}\pastdec{\alpha_z}.\underline{\alpha_z}\}\cup
  \{\cammino{m}\plusdec{z}\pastdec{\alpha_z}.\marko{m}_{Q_z}\}
  \]

 \item[r-act$\rev$] The case in
  which (\textsc{r-act$\rev$}) is used is similar.
  Assume
  \[R_1 = \mem{*}{\alpha^{z}_z}{\sum_{i\in
        I\setminus\{z\}}\alpha_i.Q_i}{m}\proc Q_z \rltsk{m}{\alpha^{}_z}
    m\proc \sum_{i\in I}\alpha_i.Q_i = R_2
  \]
  and again take $|I|\!=\!1$.
  \!Then
  $\mrkof{R_1}\!=\!\mrkof{\mem{*}{\alpha}{\nil}\!{m}\!\proc\!Q}\!=\!\mrkof{m\proc
    \alpha.Q} \setminus \{\cammino{m}\alpha.Q\}\cup
  \{\cammino{m}\pastdec{\alpha}.\underline{\alpha}\}\cup
  \{\cammino{m}\pastdec{\alpha}.\marko{m}_Q\}$. The transition
  $t = (\cammino{m}a,\re)$ in $\revnet{\ccstonet{P}}$ is enabled at
  $\mrkof{R_1}$ as it is the reverse of $(\cammino{m}\alpha,\fe)$ and its execution
  leads to the marking $\mrkof{R_2} = \mrkof{ m\proc \alpha.Q}$ as required.

  The case with $|I| > 1$ follows the same argument of the forward one.
\end{description}
  In the inductive case we have to do a case analysis on the last applied
  rule. We have (\textsc{l-par}), (\textsc{r-sych}), (\textsc{r-res}) and
  (\textsc{r-equiv}) and their reversible variants. The most representative
  cases are (\textsc{r-sych}) and (\textsc{r-equiv}).

  \setlist[description]{font=\normalfont\scshape}
    \begin{description}
\item[r-equiv]
  Consider the application of the rule (\textsc{r-equiv}). It follows by induction and by applying \cref{lem:struct}.
  \item[r-equiv$\rev$]
 The application of the rule (\textsc{r-equiv$\rev$}) follows the same argument of the previous case.

    \item[r-synch]
For the (\textsc{r-sych})
  case, let us suppose $R_0= R^1_0 \parallel R^2_0$. We have that
  $R^1_0 \parallel R^2_0 \ltsk{m_1,m_2}{\tau} R^{1}_{1_{m_1@m_2}} \parallel
  R^{2}_{1_{m_2@m_1}}$ with $R^i_0 \ltsk{m_i}{\alpha_i} R^{i}_1$ and
  $\alpha_1 = \co{\alpha_2}$. By applying the inductive hypothesis on the
  derivations $R^i_0 \ltsk{m_i}{\alpha_i} R^i_1$ we have that there exists two
  transitions $t_1$ and $t_2$ such that
  $\marko{m^i_{r_0}}\trans{t_i}\marko{m_{r_1}^i}$,
  $(\cammino{m_i}\alpha_i,\fe) = t_i$, $\mrkof{R^i_0} = \marko{m^i_{r_0}}$ and
  $\mrkof{R^i_1} = \marko{m^i_{r_1}}$. We can desume that
  $\pre{t_1} \cap \pre{t_2} = \emptyset$, since they are enabled on different
  markings. Also, by definition we have that
  $\marko{m_{r_0}}=\mrkof{R_0} = \mrkof{R^1_0} \fuse \mrkof{R^2_0}$. Let us note
  that the operator $\fuse$ acts on places which corresponds to past
  synchronisations, hence it does not affect $\pre{t_i}$, that is
  $\pre{t_i} \in \marko{m_{r_0}}$. Since $\alpha_1 = \co{\alpha_2}$ then by
  \cref{bc:parnet} in the net there exists a transition
  $t_\tau = (\{\cammino{m_1}\alpha_1,\cammino{m_2}\alpha_2\},\fe)$ where the preset
  and postset are respectively
  \(\pre{t_\tau} = \pre{t_1}\cup \pre{t_2}\) and
  \(\post{t_\tau} =( \post{t_1} \setminus
  \{\cammino{m_1}\pastdec{\alpha_1}.\underline{\alpha_1}\}) \cup ( \post{t_2}
  \setminus \{\cammino{m_2}\pastdec{\alpha_2}.\underline{\alpha_2}\})\cup
  \{s_{\{\cammino{m_1}\alpha_1,\cammino{m_2}\alpha_2\}}\}\).
   Hence we have that
  $\marko{m_{r_0}}\trans{t_\tau} (\marko{m_r}\setminus \pre{t_\tau}) \cup
  \post{t_\tau}$. By definition we have that
  $\{\cammino{m_1}\alpha_1,m_2\}\in \mrkof{R^1_1}$ and
  $\{\cammino{m_2}\alpha_2,m_1\}\in \mrkof{R^2_1}$ and that
  $\{\cammino{m_1}\alpha_1,\cammino{m_2}\alpha_2\} \in
  \mrkof{R^1_1}\fuse\mrkof{R^2_1}$. Also let us note that the $m_i@m_j$ operation
  just replace the $*$ on top of the memory $m_i$ with $m_j$, which is similar
  to the $\fuse$ operator. Hence
  $\mrkof{R^1_1}\fuse\mrkof{R^2_1} = \mrkof{R^1_{1_{m_1@m_2}} \parallel R^2_{1_{m_2@m_1}}}
  = \marko{m_{r_2}}$, as desired.

  \item[r-sych$\rev$] this case is analogous to ($\textsc{r-act}\rev$). \qedhere
  \end{description}
 \end{proof}
\begin{lem}[Completeness]\label{lem:compl}
  Let $R_1$ be an RCCS coherent process and let
  $\revnet{\ccstonet{R_1}} = \macronetbis{}{}{\mrkof{R_1}}$ be the
  corresponding \rcn{}. If $\mrkof{R_1}\trans{t}\marko{m}'$, then there exists $R_2$ s.t.
  one of the following holds:
  \begin{itemize}
  \item $t = (\cammino{m}\alpha,d)$ and $R_1 \rccst{m:\alpha} R_2$ and
    $\revnet{\ccstonet{R_2}} = \macronetbis{}{}{\marko{m}'}$;
  \item $t = (\{t_1,t_2\},d)$ such that $\co{\lbl {t_1}} = \lbl {t_2}$, with
    $t_i = (\cammino{m_i}\alpha_{i},d)$,
    $\alpha_{i}\in\{\lbl {t_1},\lbl {t_2}\}$ for $i=1,2$ and
    $R_1 \rccst{m_1,m_2:\tau} R_2$ with
    $\revnet{\ccstonet{R_2}} = \macronetbis{}{}{\marko{m}'}$
  \end{itemize}
  with $d\in\{\fe,\re\}$.
\end{lem}

\begin{proof}
  If $\mrkof{R_1}\trans{t}\marko{m}'$, then
  $\mrkof{R_1}= \marko{m}_0 \cup \pre{t}$ and
  $\marko{m}' = (\marko{m}_0 \setminus \pre{t}) \cup \post{t}$. The encoding
  of $\ccstonet\cdot$ is such that each transition or place name has a unique
  form, which corresponds to a path of a CCS term, and the transitions in
  $\revnet{\ccstonet\cdot}$ are of the form $(t,d)$, where $t$ is the
  transition name of the CCS term and $d\in \{\fe,\re\}$ is the
  \emph{direction}, either forward or reverse. That is from the transition
  name $(t,d)$ we can isolate the RCCS term which can mimic the action.

  If the transition $(t,d)$ is not a synchronisation, that is $(t,d)$ is not of
  the form $(\{t_1,t_2\},d)$, then we can assume w.l.o.g. that
  $t = (\tilde{\phi}\alpha,d)$ with $\tilde{\phi}$ being a sequence of
  $\phi \in \{\pardec{i},\plusdec{i},\pastdec{\alpha},\restdec{a}\}$. Suppose
  $d$ is $\fe$. If the last decoration in $\tilde{\phi}$ has the form
  $\plusdec{j}$, that is $\tilde{\phi}=\tilde{\phi}'\plusdec{j}$ this means
  there exists in the net a set of transition
  $T' = \{t_i = (\tilde{\phi}\beta_i,\fe) \st (t_i,\fe) \in T\}$. Now, let
  assume that the ancestor of $R_1$ is $P$, we have that
  $P = \ctx{C}{\sum_{i\in I} \beta_i.Q_i}$ where there exists an index
  $j\in I$ such that $\beta_j = \alpha$ and $\alpha$ is the action mimicked by
  the transition $(t,\fe)$ and the right position of the hole in the context
  is calculated using $\tilde{\phi}$.
  Also, since the transition is enabled in the net, then also
  $R_1 = \ctx{E}{(m\proc \sum_{i\in I} \beta_i.Q_i)\hide{A}}$ where
  $\ctx{E}{\cdot}$ is an active context. Hence, we have that
  $$\ctx{E}{(m\proc \sum_{i\in I} \beta_i.Q_i)\hide{A}} \ltsk{m}{\beta^j}
  \ctx{E}{(\mem{*}{\beta^j}{\sum_{i\in I\setminus \{j\} }}{m}\proc
    Q_j)\hide{A}} = R_2$$ By definition \ref{de:rccsasun} we have
  $\marko{m} = \mrkof{R_1}$, and by definition \ref{bc:choicenet}
  $\pre{t} = \{\tilde{\phi}\beta_i.Q_i \st i \in I\} $ and
  $\post{t} = \{\tilde{\phi}\pastdec{\beta_j}.\underline{\beta_j}\} \cup
  \tilde{\phi}.\{\beta_j.Q_j\}$. Also
  \begin{align*}
    \mrkof{R_1} &
      = \mrkof{\ctx{E}{\nil}} \fuse \mrkof{(m\proc \sum_{i\in I} \beta_i.Q_i)\hide{A}}
      = \marko{m} \cup \marko{m}_1
    \\
    \mrkof{R_2} &
      = \mrkof{\ctx{E}{\nil}} \fuse \mrkof{(\mem{*}{\beta^j}{\sum_{i\in I\setminus \{j\} }}{m}\proc Q_j)\hide{A}}
      = \marko{m} \cup \marko{m}_2
  \end{align*}
  where $\marko{m_1}$ and $\marko{m_2}$ are the results of applying the
  eventual synchronisation $\fuse$ respectively on
  $\mrkof{(m\proc \sum_{i\in I} \beta_i.Q_i)\hide{A}}$ and
  $\mrkof{(\mem{*}{\beta^j}{\sum_{i\in I\setminus \{j\} }}{m}\proc
    Q_j)\hide{A}}$. Moreover, we can separate from $\marko{m}_1$ and
  $\marko{m}_2$ the key places, that is the places whose name terminates with
  $\pastdec{\alpha}.\underline{\alpha}$ or with $s_{\{t_1,t_2\}}$. Be
  $\marko{m}^k_i$ such markings then we have:
  \begin{align*}
    &\mrkof{R_1}
      =  \marko{m} \cup \marko{m}_1
      = \marko{m} \cup \marko{m}^k_1 \cup \marko{m}'_1\\
    &\mrkof{R_2}
      = \marko{m} \cup \marko{m}_2
      = \marko{m} \cup \marko{m}^k_2 \cup \marko{m}'_2
  \end{align*}

  By definition of $\mrkof{\cdot}$ we have that

  \begin{align*}
    &\marko{m}'_1
      = \{\cammino{m}.\plusdec{i}\beta_i.Q_i \st i\in I\} \\
    &\marko{m}'_2
      = \{\cammino{(\mem{*}{\beta^j}{\sum_{i\in I\setminus \{j\} }}{m}}.\pastdec{\beta_j}.\underline{\beta_j}\} \cup \{\cammino{(\mem{*}{\beta^j}{\sum_{i\in I\setminus \{j\} }}{m}}.\pastdec{\beta_j}.Q_j\}
  \end{align*}
  It is easy to check that $\tilde{\phi} = \cammino{m}\plusdec{i}$ and
  $\tilde{\phi} = \cammino{(\mem{*}{\beta^j}{\sum_{i\in I\setminus \{j\}
      }}{m}}$. And we are done.

  The cases of synchonisation and backward transitions are similar.
\end{proof}
\noindent 
We can now state our main result in terms of bisimulation:

\begin{thm} Let $R$ be an RCCS process and let $P = \ancestor{R}$ be its
  ancestor, then \[\emp\proc P \sim_{FR} \revnet{\ccstonet{P}}\]
\end{thm}

\begin{proof} It is sufficient to show that
  \[ \mathcal{R} = \{(R, \langle S,T, F, \mrkof R\rangle) \ |\ \ancestor R =
    P, \ \revnet{\ccstonet{P}} = \langle S,T, F, \marko m\rangle \} \] is a
  forward and reverse bisimulation. It is easy to check that all the conditions of
  Definition~\ref{def:bisim} are matched by Lemmas~\ref{lem:sound} and~\ref{lem:compl}.
\end{proof}

\section{Conclusions and future works}\label{sec:concl}

On the line of previous research we have equipped a reversible process
calculus with a non sequential semantics by using one of the classical
encoding of process calculi into nets.
What comes out from the encoding is that the machinery to reverse a process
was already present in the encoding.
Other approaches to address true concurrency in reversible calculi have been
explored, for instance \cite{aubert22,aubert23}, where a proved semantics~\cite{prooved}
for CCSK is given. This requires to revisit the LTS of CCSK in order to add
extra information in the labels, about the process which contributed to an
action, and then to derive a true-concurrent notion. Our approach directly
compiles RCCS into a truly concurrent model, and hence we do not need to
modify the lts of RCCS. Hence we exploit the natively truly concurrent
semantics of Petri nets in order to retrieve a truly concurrent semantics of
RCCS. Also our approach accounts for infinite behaviours, while
\cite{aubert22,aubert23} do not. 

The current results applies to RCCS, but we do believe that the same encoding
could be used to model CCSK processes. As a matter of fact, in CCSK the
information is stored directly in the process and executed prefixes are marked
with communications keys and in our encoding it is signalled by a token in
\emph{key}-places. For example if we take the process $P = a.Q$ in CCSK the
process evolves in $a[i].Q$ where the forward behaviour of the process is $Q$
while the backward behaviour is represented by the marked prefix $a[i]$. The
same mechanisms applies to synchronisations. If we take the process
$a.b.\nil \parallel \co{a}.\nil$ the process can make a synchronisation
followed by the $b$ action and evolves to $a[i].b[j].\nil \parallel
\co{a}[i].\nil$. In this way, the synchronisation on $a$ cannot be undone if
first the action $b$ is undone. By looking on how history information is kept
into CCSK processes, it is clear that there is a tight correspondence between
the marked prefixes, the key-places and $\pastdec{\cdot}$ decorations we have
used in unravel nets. Also in CCSK the process structure does not change, and
the marking of the reversible net would correspond to the marked prefixes.
This seems to bring a more straightforward encoding of CCSK into Petri Nets,
where the marking can be easily retrieved from a CCSK term. Having the two encodings into Petri Net would
allow us for cross-fertilization results, in line with \cite{LaneseMM21}.
The
whole encoding and the machinery connected to it is left for future work.

Our result relies on unravel nets, that are able to represent
\emph{or}-causality. The consequence is that the same event may have different
pasts. Unravel nets are naturally related to \emph{bundle} event structures
\cite{LanForte92,LBKConcur97}, where the dependencies are represented using
\emph{bundles}, namely finite subsets of conflicting events, and the bundle
relation is usually written as $X \bundle e$.
Starting from an unravel net $\macronet{}{}$, and considering the transition
$t\in T$, the bundles representing the dependencies are $\pre{s}\bundle t$ for
each $s\in \pre{t}$, and the conflict relation can be easily inferred by the
semantic one definable on the unravel net. This result relies on the fact that
in any unravel net, for each place $s$, the transitions in $\pre{s}$ are
pairwise conflicting. The \emph{reversible} bundle structures add to the
bundle relation (defined also on the reversing events) a prevention relation,
and the intuition behind this relation is the usual one: some events, possibly
depending on the one to be reversed, are still present and they \emph{prevent}
that event to be reversed. The problem here is that in an unravel net,
differently from occurrence nets, is not so easy to determine which
transitions depend on the happening of a specific one, thus potentially
preventing it from being reversed. An idea would be to consider all the
transitions in $\post{s}$ for each $s\in \post{t}$, but it has to be carefully
checked if this is enough. Thus, which is the proper ``reversible bundle event
structure'' corresponding to the reversible unravel nets has to be answered,
though it is likely that the conditions to be posed on the prevention
relations will be similar to the ones considered in
\cite{GraversenPY18,GraversenPY2021}. Once that also this step is done, we
will have the full correspondence between reversible processes calculi and non
sequential models.

Another future works idea would be to move from reversible CCS to reversible
$\pi$-calculus \cite{LaneseMS16,CristescuKV13} by relying on the results of
\cite{BusiG09}. In \cite{BusiG09} a truly concurrent semantics of
$\pi$-calculus is given in form of Petri nets with inhibitor arcs. We could
exploit our previous results on reversibility and Petri nets with inhibitor
arcs \cite{MelgrattiMP21,MelgrattiMP23,tocl} to obtain a truly concurrent semantics
for reversible $\pi$-calculus in Petri nets with inhibitor arcs. Alternatively
we could exploit the encoding of reversible $\pi$-calculus into rigid families
(based on configuration structures), given in \cite{CristescuKV16}, and
bring it to Petri nets.

\bibliographystyle{alphaurl}
\bibliography{biblio}

\appendix

\section{Implementation}\label{sec:impl}

We describe an effective implementation of the proposed
encoding in \Haskell\footnote{The code can be accessed at \url{https://github.com/hmelgra/reversible-ccs-as-nets}.}.
The intent of this section is to provide a practical evidence that the coinductive approach to describe infinite behaviours
is effective, rather than providing a fully fledged tool. We are aware there exist other Petri net implementations in Haskel (see for example~\cite{Reinke99}), but a comparison with such tools is out of the scope of this section.

\subsection{Representation of infinite nets}
When working with an infinite data structure, a pivotal aspect is devising an
efficient strategy to traverse the pertinent section of the structure. In our
specific scenario, we prioritise the capability to identify and execute
enabled transitions within a (potentially infinite) net.
Therefore, our main objective is to identify those transitions that are
enabled at a given marking. %
For this purpose, we adopt a representation of infinite nets that facilitates
obtaining a truncated version of the net that contains all the enabled
transitions in given marking.
To maintain simplicity, we avoid explicitly representing the flow relation as
a set of pairs. Instead, we associate each transition with its preset (input
places) and postset (output places). Consequently, we rely on the following
instrumental datatype to represent transitions.
\InputHaskellD{27}{32}{Net.hs}
\noindent
The parameters \HI{t} and \HI{s} represent the types of the names of
transitions and places, respectively.
In this representation, a transition is defined by its name and two lists of
places, corresponding to its pre and postset.

Then, the datatype for nets is as follows:
\InputHaskellD{39}{44}{Net.hs}
The components of a net include a marking, denoted as \HI{netMarking}, which
is essentially a set of places. Additionally, there are two functions,
\HI{netPlaces} and \HI{netTransitions}, which map every marking to a set of
places and transitions, respectively, of a truncated, finite version of the
net. This truncated net includes all the transitions from the potentially
infinite net that are enabled in the given marking.

\begin{exa} Consider the infinite net $N$ depicted in Figure~\ref{fig:f-net-impl}.
  One potential \Haskell definition for $N$ could be \HI{nAt [0]}, utilising
  the function \HI{nAt} given in Figure~\ref{fig:f-net-impl-hs}. This function takes
  a marking of type \HI{[Int]} and returns a net with place names represented
  as integers and transitions as strings, i.e., of type \HI{Net Int String}.
  The net's definition relies on the functions \HI{p} and \HI{t}, which
  determine the truncation of the net corresponding to a given marking. It is
  important to note that, for a specific marking \HI{m}, any enabled
  transition $t$ in \HI{m} should satisfy the conditions $\pre t = \{i - 1\}$
  and $\post t = \{i\}$, where $0 < i \leq m$, and $m$ is the maximum integer in
  \HI{m}. Therefore, the function \HI{p}, which maps markings to sets of
  places, is defined as follows:

  \begin{itemize}
  \item For the empty marking, it returns an empty set of places since no
    transitions are enabled in the empty marking.
  \item For a non-empty marking \HI{m}, it generates a list containing all
    integers in the range from $0$ to the maximum value in \HI{m} plus one.
  \end{itemize}
  Similarly, the function \HI{t} creates a list of transitions, encompassing
  all those among the places in \texttt{p m}. These transitions are defined in
  such a way that  \HI{[i - 1]} represents its preset, and  \HI{[i]}
  represents its postset. The name of the $i$-th transition is denoted by $i$
  occurrences of \texttt{'a'}.

  \begin{figure}[thb]
    \centering
    \begin{subfigure}[b]{0.1\textwidth}
      \centering
      \scalebox{0.8}{{\footnotesize
\begin{tikzpicture}[scale=.7]
\tikzstyle{inhibitorred}=[o-,draw=red,thick]
\tikzstyle{inhibitorblu}=[o-,draw=blue,thick]
\tikzstyle{pre}=[<-,thick]
\tikzstyle{post}=[->,thick]
\tikzstyle{prered}=[<-,thick,draw=red]
\tikzstyle{postred}=[->,thick,draw=red]
\tikzstyle{readblue}=[-, draw=blue,thick]
\tikzstyle{transition}=[rectangle, draw=black,thick,minimum size=5mm]
\tikzstyle{revtransition}=[rectangle, draw=red!80,fill=red!30,thick,minimum size=5mm]
\tikzstyle{place}=[circle, draw=black,thick,minimum size=5mm]
\tikzstyle{placeblu}=[circle, draw=blue!80,fill=blue!20,thick,minimum size=5mm]
\node[place,tokens=1, label=left:$0$] (p1) at (1.5,5) {};
\node[place,label=left:$1$]          (p4) at (1.5,2.5) {};
\node[place,label=left:$2$]          (p5) at (1.5,0) {};
\node[transition] (a) at (1.5,3.75)  {$a$}
edge[pre] (p1)
edge[post](p4);
\node[transition] (c) at (1.5,1.25) {$aa$}
edge[pre] (p4)
edge[post] (p5);
\node[] (p6) at (1.5,-1.5)  {...}
edge[pre] (p5);
\end{tikzpicture}
}}
      \caption{$N$}
      \label{fig:f-net-impl}
    \end{subfigure}
    \hspace{.7cm}
    \begin{subfigure}[b]{0.7\textwidth}
      \centering
      \InputHaskellD{25}{33}{../test/TestPn.hs}
      \caption{\Haskell\ code}
      \label{fig:f-net-impl-hs}
    \end{subfigure}
    \caption{The \Haskell representation of a simple infinite net $N$}
    \label{fig:ex-impl}
  \end{figure}
\end{exa}

The auxiliary functions \HI{places} and \HI{transitions}, defined below, allow
us to respectively retrieve the set of places and transitions from the
truncation of net for its marking.
For instance, \HI{places (nAt [0])} returns \HI{[0,1]}, and
\HI{places (nAt [1,3])} gives \HI{[0,1,2,3,4]}.

\InputHaskellD{49}{53}{Net.hs}
Analogously, we rely on \HI{isTransition :: Eq t => t -> Net s t -> Bool} to
check that a given transition appears in the truncation of the net
\HI{n}.
Then, the following predicate \HI{isEnabled} allows to check if a given transition
is enabled on a net.

\InputHaskellD{70}{73}{Net.hs}
\noindent 
The guard \HI{isTransition t n} simply checks that \HI{t} appears in the
truncation of the net \HI{n}. In such case, a transition \HI{t} is enabled if
all elements in its preset appear in the marking of the net. Otherwise, the
transition is not enabled.

The firing of a transition straightforwardly changes the marking of the net
as expected, i.e., by removing the preset of the transition and by adding the
postset of the transition

\InputHaskellD{81}{84}{Net.hs}
\noindent 
Note that the firing generates an error if the transition is not enabled.

\subsection{Representing \CCS processes}
The datatype for representing \CCS actions is straightforwardly defined as
follows:
\InputHaskellD{16}{19}{Ccs.hs}
\noindent 
Note that the datatype is parametric with respect to the type \HI{a} of action
names.
The binary predicate \HI{dual} (shown below) tests whether a two actions
are dual, i.e., one is an input and the other is an output performed over the
same channel.
\InputHaskellD{28}{31}{Ccs.hs}

The datatype for representing \CCS processes is as follow.
\InputHaskellD{53}{61}{Ccs.hs}
\noindent 
The constructors are straightforward. For instance, \HI{CCS Char} stands for
the type of \CCS processes whose channel names are characters.
Then, the process $a.a \parallel \co{a} + b$ in Figure~\ref{fig:complex} is
defined as
\InputHaskellD{25}{26}{../test/TestCcs.hs}
\noindent 
We highlight that the datatype \HI{CCS} includes constructors for the finite
definition of infinite processes, i.e., \HI{Var} for a process variable and
\HI{Rec} for a recursive definition.
This choice is down to the facts that (i) our encoding uses \CCS processes as
the names of the elements of the generated nets; and (ii) the operational
semantics of nets is defined under the assumption that names can be
effectively compared (see details below). In order to have an equality test
for infinite terms, we opted for a finite representation.
Hence, the infinite \CCS process consisting of an infinite sequence
of inputs over the channel $a$ can be defined as follows

\InputHaskellD{28}{29}{../test/TestCcs.hs}

Process variables in CCS are now represented using the following type:

\InputHaskellD{42}{42}{../src/Ccs.hs}
\noindent 
This new type aims to enhance the parsing of strings into \HI{Ccs} instances
by implementing the \HI{Read} class (Details are omitted as they are
non-essential for the translation process).

This new type has been introduced in order to facilitate the parsing of
strings to CCS instances, by providing an instances of the class \HI{Read}.

When dealing with the finite representation of infinite processes, we need
the usual unfolding operation, which is defined in terms of the
substitution of a process variable by a process.
Substitution is given by the following
function
\InputHaskellD{74}{77}{Ccs.hs}
\noindent
whose defining equations are standard and therefore omitted.

The unfold function is as follows.
\InputHaskellD{91}{93}{Ccs.hs}
The function \HI{unfold} will be used in the definition of the encoding.

Despite we rely on the finite representation of \CCS{} processes, we remark
that the implementation of the encoding associates \textbf{infinite} nets to
recursive \CCS{} processes.

\subsection{Implementation of the encoding}

According to the encoding introduced in \cref{sec:coding}, the names of the places
and transitions of the obtained nets are (possibly) decorated \CCS{} processes.
We rely on the following datatypes introducing constructors for the names of
places and transitions.

\InputHaskellD{19}{30}{Encoding.hs}

\InputHaskellD{81}{91}{Encoding.hs}
\noindent 
The above definitions are in one-to-one correspondence with the names introduced
by the encoding of the previous Section, and self-explanatory.

We will use the predicate \HI{isKey} on place's names that determines if a
place name is a key, i.e., either \HI{PKey} or \HI{PSync} (its omitted
definition is straightforward).

\InputHaskellD{43}{43}{Encoding.hs}

The following function
\InputHaskellD{103}{103}{Encoding.hs}
\noindent
allows us to recover the label associated with a transition.

Then, the encoding function is given by
\InputHaskellD{143}{143}{Encoding.hs}

We now illustrate some of its representative defining equations.
According to \cref{bc:zeronet}, the encoding of the process $\zero$ (here
represented by \HI{Nil}) produces a net consisting of just one marked place.
We name that place \HI{Proc Nil}, i.e., the \CCS process 0.

\InputHaskellD{145}{145}{Encoding.hs}

The fact that the net is defined in terms of the constant functions
\HI{const [Proc Nil]} and \HI{const []} reflect that every finite truncation,
independently from the given marking, consists of just one place \HI{Proc Nil}
and none transition. The marking \HI{[Proc Nil]} assigns one token to the
unique place.

The encoding of a prefixed process follows \cref{bc:prefixnet}. Hence, the
encoding of $\alpha.P$ (written \HI{a :.p} in the implementation) is built on
top of the encoding of $P$, i.e., the names of the places and the transitions
appearing in the encoding of $P$ are decorated with the prefix
$\pastdec{\alpha}$. We use \HI{PPref a} for decorating a place name with the
past of action \HI{a} and similarly \HI{TPref a} for a transition name.
The following function (whose defining equations are omitted because are
uninteresting) is in charge of applying renamings to a net.

\InputHaskellD{117}{117}{Encoding.hs}

The first and second parameter correspond respectively to the renaming of
places and transitions. The third one is instrumental for mapping a marking on
the decorated names to a marking of the encoding of $P$, which is needed for
computing a truncation.
Then, the equation for the encoding of \HI{a :.p} is as follows.

\InputHaskellD[numbers=left]{147}{152}{Encoding.hs}
Note that line 3 introduces the net corresponding to the encoding of \HI{p},
with its element suitable renamed. Then, the places and transitions of the
(truncations of the) net are given by the defining equations of \HI{s} and
\HI{t}. Besides the fact that they are empty for empty markings, their
definitions mimic \cref{bc:prefixnet}. The encoding of \HI{p} is extended with
two places, one for the process (i.e., \HI{Proc (a :. p)}) and one for the key
(i.e, \HI{PKey a}), and one transition of name \HI{Act a}, whose preset is
\HI{Proc (a :. p)} and whose poset corresponds to the initial marking of the
encoding of \HI{p}, i.e., \HI{amp}, and  the new key \HI{PKey a}.

As for the illustrated cases, the remaining equations follow the corresponding
definitions in \Cref{sec:coding}.

\subsection{Reversing nets}

Reversible nets, are implemented as nets with tagged transitions: the tag
\HI{Fwd} stands for forward transitions and \HI{Bwd} are for reversing
transitions. The corresponding data type is as follows.

\InputHaskellD[numbers=left]{10}{12}{ReversibleNet.hs}
\noindent 
Then, the following function \HI{rev} takes a net and generates its reversible version.

\InputHaskellD[numbers=left]{18}{24}{ReversibleNet.hs}

Consider the network, denoted as \HI{Net s t m}, which is translated into a new net with the same sets of places and markings, represented as \HI{Net s t' m}. The set of transitions \HI{t'} in the new net is obtained by applying the following transformations to each transition \HI{Transition x y z} from the original set \HI{t}:

\begin{itemize}
\item Add a forward transition, denoted as \HI{Transition (Fwd x) y z}, to tag
  each transition in \HI{t} as forward.
\item Add the corresponding reversing transition, denoted as
  \HI{Transition (Rev x) y z}, to maintain the bidirectional nature of the net.
\end{itemize}

\subsection{Simulation}
The concepts introduced in the previous sections can now be effectively
utilised to simulate the behavior of reversible \CCS processes. To illustrate
this, let us consider the definition of the infinite \CCS process
\HI{ccs} below.
\InputHaskellD[numbers=left]{35}{42}{../test/TestCcs.hs}
This process is defined as the parallel composition of two infinite
processes, where the shared name \HI{1} is restricted.

To obtain the corresponding reversible net, we apply the encoding
followed by the reversing function, represented as \HI{rev(enc ccs)}.

Using the functions that determine the enabled transitions of net and
compute the firing of transitions, we can seamlessly implement a
simulation function to replicate the behavior of the process.

\InputHaskellD[numbers=left]{24}{24}{../app/Main.hs}

Then, the evaluation of
\InputHaskellD[numbers=left]{21}{21}{../app/Main.hs}
shows the set of enabled transitions of the obtained net, which are as follows.

\begin{HaskellD}
Enabled transitions:
1)  ->(|r:+l:2!)\1
2)  ->(|l:1?*|r:+r:1!)\1
\end{HaskellD}
The name \HI{(|r:+l:2!)\1} of the first transition indicates that it
corresponds to the output performed on channel \HI{2} by the left branch
(i.e., \HI{+l:}) of the right hand of the parallel composition (i.e.,
\HI{|r:}). Similarly, the symbol \HI{*} in the name \HI{|l:1?*|r:+r:1!)\1}
indicates that the transition corresponds to a synchronisation between the
input performed on channel \HI{1} by the left hand side of the parallel
composition (i.e., \HI{|l:}) and the output on channel \HI{1} performed by the
right branch of the right hand side of the parallel composition (i.e.,
\HI{|r:+r:}).

At this point, any of the two transitions can be fired. After firing the first one,
the obtained set of enabled transitions is the following.
\begin{HaskellD}
Enabled transitions:
1) ->(|r:+l:^2!.+l:2!)\1
2) ->(|l:1?*|r:+l:^2!.+r:1!)\1
3) <-(|r:+l:2!)\1
\end{HaskellD}
\noindent 
The first two transitions mirror the ones originally enabled; however, their names indicate that actions on the right-hand side of the parallel composition  causally depend on the preceding performed action (the prefix \HI{+l:^2!}).

In addition to these two forward transitions, there is one reversing transition that undoes  the previously executed action.


\end{document}